\documentclass[10pt, reqno, oneside]{amsart}

\usepackage{enumerate, microtype, amssymb, addlines, multibib}
\usepackage[pdftex]{color, graphicx}
\usepackage[longnamesfirst]{natbib}
\shortcites{dedeckeretal2007}
\defcitealias{ibragimovmueller2016}{Ibragimov-M\"uller}
\defcitealias{canayetal2014}{Canay-Romano-Shaikh}
\defcitealias{besteretal2014}{Bester-Conley-Hansen}
\newcites{Appendix}{Additional references}

\input{macros}

\theoremstyle{plain}
\newtheorem{theorem}{Theorem}[section]
\newtheorem{lemma}[theorem]{Lemma}
\newtheorem{proposition}[theorem]{Proposition}

\newtheorem{corollary}[theorem]{Corollary}
\theoremstyle{definition}
\newtheorem{algorithm}[theorem]{Algorithm}

\newtheorem{example}[theorem]{Example}
\theoremstyle{remark}
\newtheorem*{remarks}{Remarks}

\newcommand{\ev}{{\mathord \mathrm{E}}}%
\newcommand{\prob}{{\mathord P}}%
\DeclareMathOperator{\id}{\mathit{id}}%
\DeclareMathOperator*{\diag}{diag}%
\newcommand{\pto}{\mathchoice
	{\raisebox{.0em}{ $\overset{\mathrm{P}}{\to}$ }}
	{\raisebox{-.15em}{ $\overset{\raisebox{-.25em}{\scriptsize$\mathrm{P}$}}{\to}$ }}
	{}
	{}}

\newcommand{\xqed}[1]{%
  \leavevmode\unskip\penalty9999 \hbox{}\nobreak\hfill
  \quad\hbox{\ensuremath{#1}}}
\newcommand{\sqed}{\xqed{\square}}

\newcommand{\Sym}{\mathfrak{G}}
\newcommand{\G}{\mathfrak{G}}
\newcommand\mydots{\hbox to 1em{.\hss.\hss.}}

\newcommand{\balpha}{\bar{\alpha}}

\usepackage{xcolor}
\definecolor{umorange}{HTML}{cc6600}
\definecolor{umblue}{HTML}{587abc}
\definecolor{umgrey}{HTML}{989c97}
\definecolor{umred}{HTML}{7a121c}
\definecolor{umgreen}{HTML}{83b2a8}

\usepackage[bookmarks=false,pdfpagemode=UseNone,pdftex,colorlinks,citecolor=black,%
            filecolor=black,linkcolor=black,urlcolor=black,%
            pdfauthor=Andreas~Hagemann,pdfstartview=FitH]{hyperref}

\begin{document}

\title[Permutation inference with clusters]{Permutation inference with a finite\\ number of heterogeneous clusters}
\author[A.~Hagemann]{Andreas Hagemann
}
\address{University of Michigan, Stephen M. Ross School of Business, 701 Tappan Ave, Ann Arbor, MI 48109, USA. Tel.: +1 (734) 615-6663}
\email{\href{mailto:hagem@umich.edu}{hagem@umich.edu}}
\urladdr{\href{https://umich.edu/~hagem}{umich.edu/~hagem}}
\date{\today.
}
\thanks{I would like to thank co-editor Bryan Graham, three anonymous referees, Isaiah Andrews, Michal Koles\'ar, Aprajit Mahajan, and several seminar audiences for useful comments and discussions. All errors are my own.}

\begin{abstract}
I introduce a simple permutation procedure to test conventional (non-sharp) hypotheses about the effect of a binary treatment in the presence of a finite number of large, heterogeneous clusters when the treatment effect is identified by comparisons across clusters. The procedure asymptotically controls size by applying a level-adjusted permutation test to a suitable statistic. The adjusted permutation test is easy to implement in practice and performs well at conventional levels of significance with at least four treated clusters and a similar number of control clusters. It is particularly robust to situations where some clusters are much more variable than others. 

\end{abstract}

\maketitle\renewcommand{\bfseries}{\fontseries{b}\selectfont}

\section{Introduction}
It has become widespread practice in economics to conduct inference that is robust to within-cluster dependence. 
Typical examples of clusters are states, counties, cities, schools, firms, or stretches of time. Units within the same cluster are likely to influence one another or are influenced by the same external shocks. 
Several analytical and computationally intensive procedures such as the bootstrap are available to account for the presence of data clusters. Most of these procedures achieve consistency by requiring the number of clusters to go to infinity. Numerical evidence by \citet{bertrandetal2004}, \citet{mackinnonwebb2014}, and others suggests that this type of asymptotics often translates into heavily distorted inference in empirically relevant situations when the number of clusters is small or the clusters are heterogenous. In both situations, the overall finding is that true null hypotheses are rejected far too often. In this paper, I introduce an adjusted permutation procedure that is able to asymptotically control the size of tests about the effect of a binary treatment in the presence of finitely many large and heterogeneous clusters. The procedure applies to difference-in-differences estimation and other situations where treatment occurs in some but not all clusters and the treatment effect of interest is identified by between-cluster comparisons.

The main theoretical insight of this paper is that classical permutation inference can be adjusted to test the null hypothesis of equality of means of two finite samples of mutually independent but arbitrarily heterogeneous normal variables. This runs counter to classical permutation testing \citep{hoeffding1952}, where the data under the null are presumed to be exchangeable. The adjustment corrects the significance level of the test downwards to account for heterogeneity. I prove that this is possible for empirically relevant levels of significance if both samples consist of more than three observations. The corrections needed for all standard levels of significance are tabulated in the paper. I also show that if a random vector of interest converges weakly to multivariate normal with diagonal covariance matrix, then permutation inference remains approximately valid for that vector. To exploit this result in a cluster context, I construct asymptotically normal statistics from each cluster and then apply adjusted permutation inference to the collection of these statistics. The resulting permutation test is consistent against all fixed alternatives to the null, powerful against local alternatives, and is free of user-chosen parameters. 

The strategy of using cluster-level estimates as the basis for a test goes back at least to \citet{famamacbeth1973}, who without formal justification run $t$ tests on regression coefficients obtained from year-by-year cross-sectional regressions. 
Their approach is generalized and formalized by \citet{ibragimovmueller2010, ibragimovmueller2016}, who construct $t$ statistics from cluster-level estimates and show that for certain combinations of numbers of clusters and significance levels these statistics can be compared to Student $t$ critical values.  
The \citetalias{ibragimovmueller2016} test and the adjusted permutation test complement one another because they both rely on finite-sample inference with heterogeneous normal variables but apply to non-nested combinations of numbers of clusters and significance levels. The empirical example in this paper features a practically relevant situation where the \citetalias{ibragimovmueller2016} test does not apply but the adjusted permutation test does. If both tests apply, the Monte Carlo results in this paper indicate that neither test dominates the other in terms of power but the adjusted permutation test has clear advantages if the underlying data are heavy tailed.

Several other papers show that inference with a fixed number of clusters is possible under a variety of conditions: \citet{canayetal2014} permute the signs of cluster-level statistics under symmetry assumptions. This approach requires the parameter of interest to be identified \emph{within} each cluster and clusters therefore have to be paired in an ad-hoc manner for difference-in-differences estimation.  
This pairing has a substantial impact on the test decision and requires a large number of choices on the part of the researcher. \citet{besteretal2014} use standard cluster-robust covariance matrix estimators but adjust critical values under homogeneity assumptions on the clusters. \citet{canayetal2018} show that certain cluster-robust versions of the wild bootstrap 
can be valid under strong homogeneity assumptions with a fixed number of clusters. In sharp contrast, the test developed here does not require pairing clusters or any other decisions on the part of the researcher and applies even if the clusters are arbitrarily heterogeneous.

I will use the following notation: $1\{\cdot\}$ is the indicator function, $\min\{a,b \} = a\wedge b$, and cardinality of a set $A$ is $|A|$. The smallest integer larger than $a$ is $\lceil a \rceil$ and the largest integer smaller than $a$ is $\lfloor a \rfloor$.  
Limits are as $n\to\infty$ unless noted otherwise.

All proofs can be found in the online appendix.

\section{Permutation inference with heterogenous symmetric variables}\label{s:behrensfisher}
In this section I show that classical permutation inference can be adjusted to test for the equality of location of two finite samples of independent symmetric variables with heterogeneous scales. 
The discussion focuses on heterogeneous normal variables but several of the results apply more generally.

Suppose the random vector $X = (X_1,\dots, X_q)\in\mathbb{R}^q$ has entries $X_k = \mu_1 + \sigma_k Z_k$ for $1\leq k\leq q_1$ and $X_k = \mu_0 + \sigma_k Z_k$ for $q_1 +1 \leq k\leq q_1 + q_0 = q$, where the $Z_1,\dots,Z_q$ are iid symmetric variables. The $\sigma_k$ are not known and no estimates are assumed to be available. The number of variables $q$ is taken as fixed throughout this paper.  The goal is to construct an $\alpha$-level permutation test of the hypothesis $H_0\colon \mu_1=\mu_0$. This is a two-sample problem with ``treatment'' sample $X_1,\dots, X_{q_1}$ and ``control'' sample $X_{q_1+1},\dots, X_q$. 
The test statistic $T$ considered here is the comparison of means
\begin{equation}\label{eq:compmean}
	(x_1,\dots, x_q)\mapsto T(x) = \frac{1}{q_1}\sum_{k=1}^{q_1} x_k - \frac{1}{q_0}\sum_{k=q_1+1}^{q} x_k.
\end{equation}
No standardization \mbox{is needed}.

Let $\mathfrak{S}_q$ be the group of permutations of the set $\{1,\dots, q\}$. For $g\in\mathfrak{S}_q$, denote by $g(k)$ the value the permutation $g$ assigns to $k$ for $1\leq k\leq q$. The ``group action'' on $X$ in $\mathfrak{S}_q$ is the relabeling of the indices $gX = (X_{g(1)},\dots, X_{g(q)})$. A permutation test derives its critical values from the permutation statistics $T(gX)$. Because $x\mapsto T(x)$ is invariant to the ordering of the first $q_1$ and last $q_0$ entries of $x$, it suffices to compute the $T(gX)$ for the set of group actions with unique combinations of $g(1),\dots, g(q_1)$ and $g(q_1 + 1),\dots, g(q)$. One way of representing this set is
\begin{equation}\label{eq:combinations}
\Sym = \bigl\{g\in\mathfrak{S}_q : g(1) < \dots < g(q_1)\text{~and~}g(q_1+1) < \dots < g(q)\bigr\}. 
\end{equation}
Denote by $T^{(1)}(X,\Sym) \leq T^{(2)}(X,\Sym)\leq \cdots  \leq T^{(|\Sym|)}(X,\Sym)$
the ordered values of $T(gX)$ as $g$ varies over $\Sym$ and define critical values 
\begin{equation}\label{eq:critval}
	p \mapsto T^{p}(X,\Sym) = T^{(\lceil(1-p)|\Sym|\rceil)}(X,\Sym).
\end{equation}

Classical permutation inference operates under the null hypothesis that $X$ has the same distribution as $gX$ for all $g\in\mathfrak{S}_q$. In the present context this would be equivalent to assuming that $\mu_1 = \mu_0$ and that all $\sigma_k$ are identical under the null. An argument due to \citet{hoeffding1952} would then show that $T^{\alpha}(X,\Sym)$ could be used as the critical value for an $\alpha$-level test against the alternative $H_1\colon \mu_1> \mu_0$. 
If the null hypothesis is weakened to $H_0\colon \mu_1=\mu_0$ without restrictions on $\sigma_k$, a natural question to ask if there exists \emph{any} order statistic $j\mapsto T^{(j)}(X,\Sym)$, $\lceil(1-\alpha)|\Sym|\rceil\leq j< |\Sym|$, that can be used as a critical value for an $\alpha$-level test even if the classical permutation hypothesis $X\sim gX$ for all $g\in\mathfrak{S}_q$ fails. As I will discuss now, the answer to this question is affirmative for empirically relevant choices of $\alpha$ if $q_1$ and $q_0$ are larger than $3$.

Because $T(X)\in \{ T(gX) : g\in\Sym \}$, it is always true that $T(X) \leq T^{(|\Sym|)}(X,\Sym)$. The largest non-trivial critical value from $\{ T(gX) : g\in\Sym \}$ is therefore the second largest order statistic $T^{(|\Sym|-1)}(X,\Sym)$. The following theorem shows that the probability that $T(X)$ exceeds $T^{(|\Sym|-1)}(X,\Sym)$ is necessarily small under $H_0\colon \mu_1=\mu_0$. In fact, this probability is so small that $T(X) > T^{(|\Sym|-1)}(X,\Sym)$ is well below any standard choice of $\alpha$ for most values of $q_1$ and $q_0$. By monotonicity, the existence of a $j$ such that $\prob (T(X) > T^{(j)}(X,\Sym))\leq \alpha$ is then guaranteed. 
\begin{theorem}[Size for heterogeneous symmetric variables]\label{t:bfbound}
Let  $X = (X_1,\dots, X_q)$ with $X_k = \mu + \sigma_k Z_k$,  $1\leq k\leq q$, where $\sigma_1,\dots, \sigma_q >0$ and the $Z_1,\dots, Z_q$ are iid copies of a continuous random variable $Z$. If $Z$ and $-Z$ have the same distribution, then 
\[\sup_{\mu\in\mathbb{R}, \sigma_1,\dots,\sigma_q>0}\prob \bigl(T(X) > T^{(|\Sym|-1)}(X,\Sym)\bigr) = \frac{1}{2^{q_1\wedge  q_0}}.\] 
\end{theorem}

A byproduct of the theorem is a bound for the case where the scales $\sigma_1,\dots, \sigma_q$ are replaced by positive random variables independent of $Z_1,\dots,Z_q$. The $X_k$ are then called ``scale mixtures'' of a symmetric variable $Z$. 
The following corollary is immediately obtained from Theorem \ref{t:bfbound} by conditioning on a given set of random scales.\footnote{A referee points out that \citet{szekely2006} studies one-sample Student $t$-tests for similar classes of distributions. \citeauthor{szekely2006} does not deal with permutation inference and uses a fundamentally different proof technique but the results are also powers of two.}

\begin{corollary}[Size for symmetric scale mixtures]\label{t:symscalebound}
Suppose $X = (X_1,\dots, X_q)$ with $X_k = \mu + S_k Z_k$,  $1\leq k\leq q$, where the $Z_1,\dots, Z_q$ are iid copies of a continuous random variable $Z$ and $(S_1, \dots, S_q$) is a possibly dependent random vector independent of $Z_1,\dots, Z_q$ with $P(S_k > 0) = 1$ for $1\leq k\leq q$. If $Z$ and $-Z$ have the same distribution, then $\sup_{\mu\in\mathbb{R}}\prob (T(X) > T^{(|\Sym|-1)}(X,\Sym)) \leq 1/2^{q_1\wedge  q_0}$.
\end{corollary}

Theorem \ref{t:bfbound} shows that a test with critical value $T^{(|\Sym|-1)}(X,\Sym)$ has size $1/2^{q_1\wedge  q_0} = 0.0625$, $0.0313$, $0.0156$, $0.0078$, $0.0039$ as $q_1\wedge  q_0$ increases from $4$ to $8$. Consequently, a 10\%-level permutation test that relies only on symmetry is available with $q_1$ and $q_0$ as small as $4$. One can perform a 5\%-level test with $q_1\wedge  q_0\geq 5$, a 5\%-level two-sided test (see the discussion below \eqref{eq:adjustedtestdec} ahead) with $q_1\wedge  q_0\geq 6$, a 1\%-level test with $q_1\wedge  q_0 \geq 7$, and a 1\%-level two-sided test with $q_1\wedge  q_0 \geq 8$.

More generally, Theorem \ref{t:bfbound} implies that for many combinations of $q_1$, $q_0$, and $\alpha$ there exist $p\in (0,1)$ such that $\lceil (1-\alpha)|\Sym|\rceil \leq \lceil (1-p)|\Sym|\rceil < |\Sym|$ and $\prob (T(X) > T^{p}(X,\Sym)) \leq \alpha$. The largest such value of $p$ maximizes power while still controlling the size of the test. Finding this $p$ is theoretically and computationally challenging. However, computation can be simplified if $Z$ is restricted to a single distribution. For normal distributions, the best possible $p$ is
\begin{multline} \label{eq:balpha}
\balpha = \sup\biggl\{ p \in [0,1) : \sup_{\mu\in\mathbb{R}, \sigma_1,\dots,\sigma_q>0} \prob\bigl(T(X) > T^{p}(X,\Sym)\bigr) \leq \alpha,\\ X\sim N\bigl(\mu, \diag(\sigma^2_1,\dots, \sigma^2_q)\bigr)\biggr\},	
\end{multline}
where I suppress the dependence on $q_1$ and $q_0$ to prevent notational clutter. By construction, $\balpha$ controls the size of the permutation test not only for arbitrarily heterogeneous normal variables but also for the entire class of scale mixtures of normals. This class includes all Student $t$ and Laplace distributions, as well as many other standard distributions \citep[see, e.g.,][]{gneiting1997}. Moreover, because the critical value is from a permutation distribution, the test also controls size for all exchangeable distributions. The remainder of the paper therefore focuses on this $\balpha$ and heterogeneous normal $X$ but other choices of distributions are possible. 

A convenient feature of $\balpha$ is that it does not depend on the data and can therefore be tabulated. To this end, I use a location-scale invariance argument to reduce the inner supremum in \eqref{eq:balpha} to a supremum over $(0,1]^q$, simulate $\prob$ over large random grids on $(0,1]^q$, and compute $\balpha$ by iteratively searching over these grids. (See Online Appendix~\ref{s:balpha} for details.) 
The search is not exhaustive and does not guarantee that the target quantity in \eqref{eq:balpha} is found. However, in experiments this method consistently replicated the theoretical result in Theorem~\ref{t:bfbound} up to a small approximation error, which indicates---but does not unequivocally establish---that this approximation of $\balpha$ is reliable.

\begin{table}\caption{Values for $\balpha$ as defined in \eqref{eq:balpha} as a function of $q_1$, $q_0$, and $\alpha$.}\label{tb:balphavals}
\centering
\scalebox{.9}{
\begin{tabular}{cp{0cm}cp{0cm}cccccccccc}																								\hline
	&	&		&	&	\multicolumn{9}{c}{$q_0$}																	\\	\cline{5-13}
$\alpha$	&	&	$q_1$	&	&	4	&	5	&	6	&	7	&	8	&	9	&	10	&	11	&	12	\\	\hline
.10	&	&	4	&	&	.0428	&		&		&		&		&		&		&		&		\\	
	&	&	5	&	&	.0317	&	.0595	&		&		&		&		&		&		&		\\	
	&	&	6	&	&	.0238	&	.0432	&	.0660	&		&		&		&		&		&		\\	
	&	&	7	&	&	.0181	&	.0340	&	.0500	&	.0760	&		&		&		&		&		\\	
	&	&	8	&	&	.0161	&	.0303	&	.0493	&	.0600	&	.0813	&		&		&		&		\\	
	&	&	9	&	&	.0153	&	.0246	&	.0400	&	.0580	&	.0740	&	.0900	&		&		&		\\	
	&	&	10	&	&	.0129	&	.0220	&	.0366	&	.0500	&	.0700	&	.0826	&	.0926	&		&		\\	
	&	&	11	&	&	.0153	&	.0193	&	.0313	&	.0420	&	.0606	&	.0746	&	.0853	&	.0953	&		\\	
	&	&	12	&	&	.0106	&	.0193	&	.0260	&	.0420	&	.0580	&	.0673	&	.0800	&	.0926	&	.0953	\\	
																							\rule{0pt}{3ex}
.05	&	&	5	&	&		&	.0158	&		&		&		&		&		&		&		\\	
	&	&	6	&	&		&	.0108	&	.0227	&		&		&		&		&		&		\\	
	&	&	7	&	&		&	.0088	&	.0200	&	.0253	&		&		&		&		&		\\	
	&	&	8	&	&		&	.0062	&	.0120	&	.0233	&	.0306	&		&		&		&		\\	
	&	&	9	&	&		&	.0113	&	.0120	&	.0213	&	.0300	&	.0393	&		&		&		\\	
	&	&	10	&	&		&	.0100	&	.0113	&	.0166	&	.0286	&	.0340	&	.0420	&		&		\\	
	&	&	11	&	&		&	.0100	&	.0080	&	.0153	&	.0240	&	.0313	&	.0393	&	.0440	&		\\	
	&	&	12	&	&		&	.0073	&	.0080	&	.0153	&	.0213	&	.0266	&	.0366	&	.0440	&	.0491	\\	
																							\rule{0pt}{3ex}	
.025	&	&	6	&	&		&		&	.0043	&		&		&		&		&		&		\\	
	&	&	7	&	&		&		&	.0040	&	.0086	&		&		&		&		&		\\	
	&	&	8	&	&		&		&	.0026	&	.0086	&	.0153	&		&		&		&		\\	
	&	&	9	&	&		&		&	.0026	&	.0066	&	.0100	&	.0146	&		&		&		\\	
	&	&	10	&	&		&		&	.0026	&	.0046	&	.0093	&	.0146	&	.0166	&		&		\\	
	&	&	11	&	&		&		&	.0020	&	.0033	&	.0080	&	.0106	&	.0166	&	.0180	&		\\	
	&	&	12	&	&		&		&	.0020	&	.0033	&	.0073	&	.0093	&	.0120	&	.0173	&	.0206	\\	
																							\rule{0pt}{3ex}	
.01	&	&	7	&	&		&		&		&	.0026	&		&		&		&		&		\\	
	&	&	8	&	&		&		&		&	.0013	&	.0026	&		&		&		&		\\	
	&	&	9	&	&		&		&		&	.0013	&	.0020	&	.0033	&		&		&		\\	
	&	&	10	&	&		&		&		&	.0013	&	.0020	&	.0033	&	.0040	&		&		\\	
	&	&	11	&	&		&		&		&	.0013	&	.0020	&	.0033	&	.0040	&	.0066	&		\\	
	&	&	12	&	&		&		&		&	.0013	&	.0013	&	.0026	&	.0033	&	.0053	&	.0066	\\	
																							\rule{0pt}{3ex}	
.005	&	&	8	&	&		&		&		&		&	$*$	&		&		&		&		\\	
	&	&	9	&	&		&		&		&		&	$*$	&	.0013	&		&		&		\\	
	&	&	10	&	&		&		&		&		&	$*$	&	.0013	&	.0013	&		&		\\	
	&	&	11	&	&		&		&		&		&	$*$	&	.0006	&	.0013	&	.0020	&		\\	
	&	&	12	&	&		&		&		&		&	$*$	&	$*$	&	.0013	&	.0020	&	.0033	\\	\hline
\multicolumn{13}{p{.85\textwidth}}{\footnotesize\emph{Note:} $*$ means $T^{\balpha}(X,\G)$ should be the second largest order statistic $T^{(|\G|-1)}(X,\G)$. More values are available at \href{https://hgmn.github.io/ap/}{\texttt{https://hgmn.github.io/ap}}.}																
\end{tabular}																								
}			
\end{table}

 Table \ref{tb:balphavals} lists $\balpha$ for common choices of $\alpha$ as a function of $q_1$ and $q_0$. As can be seen, the adjustment needed to make inference robust to variance heterogeneity is substantial if $q_1\wedge q_0$ is very small but disappears quickly as $q_1\wedge q_0$ increases. For example, for $q_1 = 4 = q_0$ a robust 10\%-level test requires using the 95.62\% quantile of the unadjusted test but for $q_1 = 9 = q_0$ the 91\% quantile is already sufficient for a robust 10\%-level test. For larger numbers of variables the need for adjustment nearly disappears at conventional levels of significance. This is also confirmed by results in \citet{hagemann2019}, who shows that unadjusted permutation inference in this context with the statistic $T(X)$ is consistent if the number of treated and control units grows in a balanced manner.

The test decision is now simple. For $q_1 \wedge q_0 > 3$, choose $\balpha$ for a feasible $\alpha$ from Table~\ref{tb:balphavals} to ensure $\prob (T(X) > T^{\balpha}(X,\Sym))\leq \alpha$ under $H_0\colon \mu_1 = \mu_0$. The existence of such an $\balpha$ for the comparison-of-means test statistic $T$ is guaranteed by Theorem~\ref{t:bfbound}. For an $\alpha$-level test of the null hypothesis $H_0\colon \mu_1 = \mu_0$, reject in favor of the alternative $H_1\colon \mu_1 > \mu_0$ if 
\begin{equation} \label{eq:adjustedtestdec}
	T(X) >  T^{\balpha}(X,\Sym).
\end{equation} 
For a one-sided test of level $\alpha$ against $\mu_1 < \mu_0$, reject if $T(-X) >  T^{\balpha}(-X,\Sym)$ or, equivalently, $T(X) <  T^{(\lfloor |\Sym|\balpha \rfloor)}(X,\Sym)$. For a two-sided test of level $2\alpha$ against $\mu_1 \neq \mu_0$, reject if $T(X) >  T^{\balpha}(X,\Sym)$ or $T(-X) >  T^{\balpha}(-X,\Sym)$. Test decisions can also be equivalently made with the \emph{p}-value of the unadjusted test 
\begin{equation}\label{eq:pval}
\hat{p}(X, \Sym) =  \inf\{ p \in (0,1) : T(X) >  T^{p}(X,\Sym) \} = \frac{1}{|\Sym|}\sum_{g\in\Sym}1\{ T(gX) \geq T(X) \}
\end{equation}
because $T(X) >  T^{p}(X,\Sym)$ if and only if $\hat{p}(X, \Sym) \leq p$ for every $p \in (0,1)$.
A $p$-value for a two-sided test can be defined as
	$2( \hat{p}(X, \Sym) \wedge \hat{p}(-X, \Sym) ).$ Reject the null hypothesis if the $p$-value does not exceed $\balpha$ from Table \ref{tb:balphavals} to perform an $\alpha$-level test.

Online Appendix \ref{s:moretheory} contains additional results on power, stochastic approximation of $\G$, and large sample approximation of $X$. The next section applies Theorem \ref{t:bfbound} to situations where $X$ is the distributional limit of cluster-level statistics.

\section{Permutation inference with heterogenous clusters}\label{s:clusterperm}
In this section, I establish large sample results for an adjusted permutation test with finitely many clusters under a single high-level condition. I then outline how these results can be applied in empirical practice.

\addline

Suppose data from $q$ large clusters (e.g., counties, regions, schools, firms, or stretches of time) are available. Observations are independent across clusters but dependent within clusters. An intervention took place during which clusters $1\leq k\leq q_1$ received treatment and clusters $q_1 + 1\leq k\leq q$ did not. The quantity of interest is a treatment effect or an object related to a treatment effect that can be represented by a scalar parameter~$\delta$. Because entire clusters receive treatment, this parameter is only identified up to a location shift $\theta_0$ within a treated cluster. Hence, only the left-hand side of \[ \theta_1 = \theta_0 + \delta \] can be identified from such a cluster. If the clusters have similar characteristics, then $\theta_0$ can be identified from an untreated cluster. Comparing the two clusters identifies $\delta$.

The identification strategy outlined in the preceding paragraph is the basis for differences-in-differences estimation---arguably the most popular identification strategy in economics today---and a variety of other models. The purpose of this section is to use the results from Section \ref{s:behrensfisher} to develop a permutation test of the conventional (non-sharp) hypothesis 
\[H_0\colon \delta = 0, \] 
or, equivalently, $H_0\colon \theta_1 = \theta_0$. The idea is to obtain independent estimates \smash{$\hat{\theta}_{n,1},\dots, \hat{\theta}_{n,q_1}$} of $\theta_1$ and independent estimates \smash{$\hat{\theta}_{n,q_1 + 1},\dots, \hat{\theta}_{n,q}$} of $\theta_0$ so that $\hat{\theta}_n = (\hat{\theta}_{n,1},\dots, \hat{\theta}_{n,q})$ is approximately multivariate normal with diagonal covariance matrix. 
The following example outlines a simple situation where this is possible. 

\begin{example}[Difference in differences]\label{ex:diffindiff}
Consider the regression model
\begin{equation}\label{eq:diffindiff}
Y_{t,k} = \theta_0 I_t + \delta I_t D_{k} + \beta_k' X_{t,k} + \zeta_k + U_{t,k},
\end{equation}
where $k$ indexes individual units, $t$ indexes time, $I_t=1\{t>n_{0,k}\}$ indicates time periods after an intervention at a known time $n_{0,k}$, the dummy $D_k$ indicates whether unit $k$ eventually received treatment, and the $\zeta_k$ are individual fixed effects. 
Provided $U_{t,k}$ has conditional mean zero and the covariates $X_{t,k}$ vary before or after $n_{0,k}$, the data identify $\theta_1 = \theta_0 + \delta$ in a treated cluster and $\theta_0$ in an untreated cluster. View each cluster as a separate regression and rewrite \eqref{eq:diffindiff} as
\begin{equation}\label{eq:diffindiffsep}
Y_{t,k} = 
\begin{cases} 
\theta_1 I_t + \beta_k' X_{t,k} + \zeta_k + U_{t,k}, &1\leq k\leq q_1, \\ 
\theta_0 I_t + \beta_k' X_{t,k} + \zeta_k + U_{t,k}, &q_1< k\leq q
\end{cases}
\end{equation}
and use the least squares estimates $\hat{\theta}_{n,k}$ of $\theta_1$ and $\theta_0$ as $\hat{\theta}_n = (\hat{\theta}_{n,1},\dots, \hat{\theta}_{n,q})$. \sqed
\end{example}

\vspace*{-\baselineskip}

The cluster-level statistics $\hat{\theta}_n$ can be combined with the results in the previous section to perform a consistent permutation test as the sample size $n$ grows large. The test is not limited to the $\hat{\theta}_n$ constructed in the preceding example. Instead, the key high-level condition is that a centered and scaled version of some estimate $\hat{\theta}_n$ converges to a $q$-dimensional standard normal distribution,
\begin{equation}\label{eq:jointconv}
\sqrt{n}\biggl(\frac{\hat{\theta}_{n,1} - \theta_1}{\sigma_1(\theta_1)}, \dots, \frac{\hat{\theta}_{n,q_1} - \theta_1}{\sigma_{q_1}(\theta_1)},\frac{\hat{\theta}_{n,q_1+1} - \theta_0}{\sigma_{q_1+1}(\theta_0)}, \dots, \frac{\hat{\theta}_{n,q} - \theta_0}{\sigma_{q}(\theta_0)}\biggr) \leadsto N(0, I_q).	
\end{equation}
The $\sigma_1,\dots, \sigma_q$ may depend on $\theta_1$ or $\theta_0$ but are not presumed to be known or estimable by the researcher. 
This is an important feature of the test because consistent covariance matrix estimation would require knowledge of an explicit ordering of the dependence structure within each cluster. While ordering the data is straightforward for time-dependent data, it may be difficult or impossible to infer or credibly assume an ordering of the data within villages or schools. In contrast, \eqref{eq:jointconv} can be established under weak dependence assumptions where it is only presumed that there \emph{exists} a possibly unknown ordering for which the dependence decays at a certain rate. 
\citet{machkouriaetal2013} present easy-to-use moment bounds and limit theorems for this situation; see also \citet{besteretal2014} for further results.

I now show that under the joint convergence \eqref{eq:jointconv} a permutation test based on comparison of means of $\smash{\hat{\theta}_{n,1},\dots, \hat{\theta}_{n,q_1}}$ and \smash{$\hat{\theta}_{n,q_1 + 1},\dots, \hat{\theta}_{n,q}$} can be adjusted to be asymptotically of level $\alpha$ with a fixed number of clusters. This is possible for $q_1 \wedge q_0 > 3$ if $\balpha$ in Table \ref{tb:balphavals} is available at the desired significance level $\alpha$. In that case, the test has power against fixed alternatives $\theta_1 = \theta_0 + \delta$ with $\delta > 0$ and local alternatives $\theta_1 = \theta_0 + \delta/\sqrt{n}$ converging to the null. In the latter situation, $\theta_0$ is fixed and $\theta_1$ implicitly depends on $n$. The convergence in \eqref{eq:jointconv} is then no longer pointwise in $\theta = 
(\theta_1,\theta_0)$ but a statement about the sequence $\theta_n = (\theta_0 + \delta/\sqrt{n}, \theta_0)$. As before, the test can be made two-sided to have power against fixed and local alternatives from either direction. Let $x\mapsto\tilde{\Phi}_{\theta_0}(x) = \prod_{1\leq k\leq q_0}\Phi(x/\sigma_{k+q_1}(\theta_0))$.
\begin{theorem}[Consistency and local power]\label{t:behrensfisherasy}
Suppose \eqref{eq:jointconv} holds.  
If $\theta_1 = \theta_0$, then \[\lim_{n\to\infty}\prob_\theta \bigl( T(\hat{\theta}_n) >  T^{\balpha}(\hat{\theta}_n,\Sym)\bigr) \leq \alpha.\]
Let $\balpha \geq 1/|\Sym|$. If $\theta_1 = \theta_0 + \delta$ with $\delta > 0$, then $\prob_\theta ( T(\hat{\theta}_n) >  T^{\balpha}(\hat{\theta}_n,\Sym)) \to 1$. If $\theta_1 = \theta_0 + \delta/\sqrt{n}$ and the $\sigma_1,\dots, \sigma_{q}$ are continuous and positive at $\theta_0$, then 
\begin{equation}\label{eq:asylocalpower}
\lim_{n\to\infty}\prob_{\theta_n} \bigl(T(\hat{\theta}_n) > T^{\balpha}(\hat{\theta}_n,\Sym)\bigr)  \geq \int_0^1 \prod_{1 \leq j \leq  q_1} \Phi \Biggl( \frac{\delta - \tilde{\Phi}^{-1}_{\theta_0}(t)}{\sigma_j(\theta_0)} \Biggr)dt.
\end{equation}
\end{theorem}

\begin{remarks}
(i)~Because $T(\hat{\theta}_n) >  T^{\balpha}(\hat{\theta}_n,\Sym)$ if and only if $T(a(\hat{\theta}_n - \theta_0 1_q)) >  T^{\balpha}(a(\hat{\theta}_n - \theta_01_q),\Sym)$, where $a> 0$ and $1_q$ is a $q$-vector of ones, the root-$n$ rate in \eqref{eq:jointconv} and in the theorem can be replaced by any other rate as long as the asymptotic normal distribution in \eqref{eq:jointconv} is still attained. The theorem therefore covers several semiparametric or nonstandard estimators.

(ii)~To test $H_0\colon \theta_1 = \theta_0 + \lambda$ for a given $\lambda$, define $\Lambda = (\lambda 1\{k\leq q_1 \})_{1\leq k\leq q}$ and reject if $T(\hat{\theta}_n - \Lambda) >  T^{\balpha}(\hat{\theta}_n - \Lambda,\Sym)$. 
Consistency follows from part~(i) of this remark and Theorem~\ref{t:behrensfisherasy}.

(iii)~If evaluating all elements of $\Sym$ is too costly, the computational burden can be reduced by working with a random sample $\Sym_m$ of $m$ random draws from $\Sym$. As long as $m\to\infty$ and then $n\to\infty$, the theorem and parts (i)-(ii) of this remark also hold for $\Sym_m$ with the exception of the local power bound if $\balpha|\Sym|$ happens to be an integer. In that case, the inequality \eqref{eq:asylocalpower} holds after subtracting $\prob(\hat{p}(Y, \Sym) = \balpha)/2$ from its right-hand side, where $Y = (\sigma_1(\theta_0)Z_1, \dots, \sigma_q(\theta_0)Z_q)$, the $Z_1,\dots, Z_q$ are independent standard normal, and $\hat{p}$ is defined in \eqref{eq:pval}. This corrects for the discreteness of the test. (See also Online Appendix \ref{s:moretheory}.)  \sqed 

\end{remarks}

\begin{example}[Difference in differences, cont.]\label{ex:diffindiffcont}
Suppose there are $n_{0,k}$ pre-intervention and $n_{1,k}$ post-intervention periods for unit $k$. The data from the $n_k = n_{0,k} + n_{1,k}$ time periods available for unit $k$ are the $k$-th cluster. Let $n = \sum_{k=1}^q n_k$.  In the absence of covariates (i.e., $\beta_k\equiv 0$), each least squares estimate in \eqref{eq:diffindiffsep} satisfies 
\[ \sqrt{n}(\hat{\theta}_{n,k} - \theta_0) = \biggl(\frac{n}{n_{1,k}}\biggr)^{1/2} n_{1,k}^{-1/2}\sum_{t=n_{0,k}+ 1}^{n_{k}}U_{t,k} - \biggl(\frac{n}{n_{0,k}}\biggr)^{1/2} n_{0,k}^{-1/2}\sum_{t=1}^{n_{0,k}}U_{t,k} \] 
under $H_0$. If the pre-intervention and post-intervention periods are long in the sense that $n/n_{0,k} \to c_{0,k} \in (0,\infty)$ and $n/n_{1,k} \to c_{1,k} \in (0,\infty)$ for $1\leq k\leq q$, then condition \eqref{eq:jointconv} already holds if $n^{-1/2}(\sum_{t=1}^{n_{0,k}}U_{t,k},$ $\sum_{t=n_{0,k}+ 1}^{n_{k}}U_{t,k})$ is independent across $1\leq k\leq q$ and has a non-degenerate normal limiting distribution for each $k$. A large number of central limit theorems for time dependent data exist; see, e.g., \citet{white2001}. 
Alternatively, if relatively few post-intervention periods are available so that $n_1 = \sum_{k=1}^q n_{1,k}$ satisfies $n_{1}/n_{0,k} \to 0$ and $n_1/n_{1,k} \to c_{1,k} \in (0,\infty)$ for $1\leq k\leq q$, the scale invariance of the test allows replacement of the $\sqrt{n}$ in \eqref{eq:jointconv} by $\sqrt{n_1}$. Then \eqref{eq:jointconv} holds if $n_{0,k}^{-1/2}\sum_{t=1}^{n_{0,k}}U_{t,k} = O_P(1)$ and $n_{1,k}^{-1/2}\sum_{t=n_{0,k}+ 1}^{n_{k}}U_{t,k}$ obeys a central limit theorem for $1\leq k\leq q$. This argument also applies if relatively few pre-intervention periods are available with the roles of $n_{0,k}$ and $n_{1,k}$ reversed. If the pre-intervention and post-intervention periods are short, Theorem \ref{t:bfbound} implies that the permutation test can still be applied if $(U_{t,k})_{1\leq t\leq n_k}$ is multivariate normal for $1\leq k\leq q$. 
 
 The calculations in the preceding paragraph can be adjusted to include covariates. 
 Similar calculations also apply if each cluster is a collection of individual-level data over time, 
 although in that case more general limit theory is needed. See, e.g, \citet{jenischprucha2009} and \citet{machkouriaetal2013} for appropriate results. 
 
The model in \eqref{eq:diffindiff} can be modified in several ways. For instance, cluster-specific $\delta_k$ can be assumed instead of a fixed $\delta$. The null hypothesis is then $\delta_k = 0$ for all $k\leq q_1$ and the test has power against the alternative $\min_{k\leq q_1} \delta_k > 0$ without changes to estimation and inference. (Conversely, the parameter $\beta_k$ does not need to vary across clusters for the results to go through.) The method discussed here can also be applied in difference-in-difference designs with staggered adoption  \citep[see, e.g.,][]{chaisemartindhaultfoeille2020}. However, as \citet{rothetal2022} point out, $\theta_0$ cannot vary by cluster, which rules out heterogeneous trends in untreated potential outcomes across clusters. 
 \sqed	
\end{example}

Online Appendix \ref{s:moreexamples} provides more practical guidance for the implementation of the adjusted permutation test and applies the test to several standard econometric models.

\section{Numerical results}\label{s:montecarlo}

This section studies the behavior of the adjusted permutation test and related methods 
in a Monte Carlo experiment and in data from a randomized trial. The discussion focuses on one-sided tests to the right but the results apply more generally. Online Appendix \ref{s:morenumerical} contains additional numerical examples and empirical applications.

\begin{example}[Difference in differences, cont.]\label{ex:diffindiffnum} This example explores the behavior of the adjusted permutation (AP hereafter) test, the \citet[IM]{ibragimovmueller2016} test (see Online Appendix \ref{s:morenumerical} for a description and more results), the \citet[BCH]{besteretal2014} test, and a clustered wild bootstrap \citep[WCB]{cameronetal2008} in a version of a Monte Carlo experiment in \citet{conleytaber2011}. The BCH test estimates parameters by least squares in the pooled sample and standardizes this estimate with the usual cluster-robust covariance matrix with a degrees-of-freedom adjustment.  
The resulting statistic is compared to the $1-\alpha$ quantile of $t$ distribution with $q-1$ degrees of freedom. BCH show that this test is valid for certain ranges of $q$ and $\alpha$ under regularity conditions if the distribution of the covariates is very similar across clusters. The WCB takes the same statistic but compares it to the bootstrap distribution of the statistic obtained from the cluster-robust version of the wild bootstrap using the Rademacher distribution and with the null  imposed. This procedure is outlined in detail in \citet{cameronetal2008}. It is valid with $q\to \infty$ \citep{djogbenouetal2019} under mild homogeneity conditions and valid for fixed $q$ under strong homogeneity conditions \citep{canayetal2018}. The bootstrap here uses 199 repetitions.

The data generating process is the model in \eqref{eq:diffindiff} specialized to
\begin{gather*}
Y_{t,k} = \theta_0 I_t + \delta I_t D_{k} + \beta_1 X_{1, t,k} + \beta_2 X_{2, t,k} + \beta_3 X_{3, t,k} + \zeta_k + U_{t,k},\\	
U_{t, k} = \rho U_{t-1, k} + V_{t,k}, \qquad X_{1, t,k} = \gamma I_t D_{k} + W_{t,k},
\end{gather*}
with $\theta_0 = \beta_1 = \beta_2 = \beta_3 = 1$, $\zeta_k\equiv 1$, $\rho  = 0.5$, and $\gamma = 0.8$. As before, $I_t = 1\{ t > n_{0,k} \}$ is a post-intervention indicator and $D_k$ is a treatment indicator. There are $n_{0,k} \equiv 10$ pre-intervention and $n_{1,k}\equiv 10$ post-intervention periods, six clusters received treatment, and six did not. 
I consider $(X_{2,t,k}, X_{3,t,k}, V_{t,k}, W_{t,k}) \sim N(0, \sigma_k^2 \mathrm{I})$ for every $1\leq k\leq q$ and $t$. The experiment varies $\delta \in \{0, 1, 2, 3\}$ and cluster heterogeneity $h$ as follows: for $h\in\{ 1, 3, 5, 7 \}$, the last $h$ clusters had $\sigma_{q-h+1} = \dots = \sigma_q = 20$ and the remaining $q-h$ clusters had $\sigma_{1} = \dots = \sigma_{q-h} = 1$. 

\begin{table}\caption{Rejection frequencies of the adjusted permutation test (AP) test, \citetalias{ibragimovmueller2016} (IM) test, \citetalias{besteretal2014} (BCH) test, wild cluster bootstrap (WCB), and an oracle version of the \citetalias{canayetal2014} (CRS) test for increasing degrees of heterogeneity $h$ in Example~\ref{ex:diffindiffnum}.}\label{tb:did}
\scalebox{.9}{
\begin{tabular}{cp{0cm}cccccp{0cm}ccccc}																								\hline	
	&	&	\phantom{oracle}	&	\phantom{oracle}	&	\phantom{oracle}	&	\phantom{oracle}	&	oracle	&	&	\phantom{oracle}	&	\phantom{oracle}	&	\phantom{oracle}	&	\phantom{oracle}	&	oracle	\\		
	&	&	AP	&	IM	&	BCH	&	WCB	&	CRS	&	&	AP	&	IM	&	BCH	&	WCB	&	CRS	\\	\cline{3-7} \cline{9-13}	
$h$	&	&	\multicolumn{5}{c}{$\delta = 0$ (size)}									&	&	\multicolumn{5}{c}{$\delta = 1$ (power)}									\\	\cline{1-1}\cline{3-7} \cline{9-13}	
1	&	&	.0244	&	.0086	&	.0265	&	.0392	&	.0474	&	&	.2826	&	.1176	&	.2930	&	.3981	&	.4437	\\		
3	&	&	.0316	&	.0287	&	.0641	&	.0538	&	.0513	&	&	.1214	&	.0706	&	.1433	&	.1493	&	.1627	\\		
5	&	&	.0377	&	.0507	&	.0787	&	.0635	&	.0451	&	&	.0549	&	.0662	&	.1086	&	.0887	&	.0792	\\		
7	&	&	.0358	&	.0475	&	.0735	&	.0634	&	.0442	&	&	.0438	&	.0560	&	.0924	&	.0791	&	.0659	\\		
\rule{0pt}{3ex}		
	&	&	\multicolumn{5}{c}{$\delta = 2$ (power)}									&	&	\multicolumn{5}{c}{$\delta = 3$ (power)}									\\	\cline{3-7} \cline{9-13}	
1	&	&	.5541	&	.3142	&	.5631	&	.6234	&	.6036	&	&	.6227	&	.4797	&	.7001	&	.7054	&	.6799	\\		
3	&	&	.1896	&	.1263	&	.2375	&	.2435	&	.2410	&	&	.2445	&	.1900	&	.3448	&	.3420	&	.3056	\\		
5	&	&	.0728	&	.0889	&	.1566	&	.1325	&	.1192	&	&	.0982	&	.1188	&	.2214	&	.1897	&	.1565	\\		
7	&	&	.0533	&	.0707	&	.1306	&	.1110	&	.0908	&	&	.0715	&	.0915	&	.1715	&	.1488	&	.1168	\\	\hline	
\end{tabular}																									}																	
\end{table}

Table \ref{tb:did} shows the rejection frequencies of the four tests outlined above under the null and the alternative. Each entry was computed from 10,000 Monte Carlo simulations and all methods were faced with the same data. As can be seen, all tests were conservative when there was little heterogeneity ($h=1$). However, the BCH test and the WCB were no longer able to control size as the heterogeneity increased. The over-rejection in both methods led to higher rejection frequencies under the alternative, which therefore should not be viewed as evidence of their power.  
The AP test rejected far more false nulls than the IM test when there was little heterogeneity. As the heterogeneity increased, the IM test had a slight advantage. The BCH test and the WCB performed well at $h=1$. However, even then there was little cost to using the AP test. It rejected nearly as many false nulls as the BCH test and at most 11.55 percentage points fewer false nulls than the WCB but was able to control size.

Several other methods for inference specifically designed for difference in differences such as \citet{donaldlang2007} and \citet{conleytaber2011} are available. Here I focus only on methods that apply more broadly and that are valid with a fixed number of clusters. The test of \citet[CRS]{canayetal2014} technically applies here but requires matching each treated cluster with a control cluster. In the present example, there are $6! = 720$ potential matches and equally many potential tests. A single match is enough to perform the test but different matches can lead to different test outcomes. This arbitrariness can be unattractive in applied work because the number of ways in which tests can be selected (and potentially combined) is large. However, if a pilot study or pre-analysis plan prescribed the cluster pairs, the (randomized) CRS test would be asymptotically similar and therefore provides a useful benchmark for the AP test. To this end, Table~\ref{tb:did} shows results of an oracle version of the CRS test that presumes that a pre-analysis plan is in place. As can be seen, the AP test compares well to the CRS test while completely avoiding the issue that different cluster pairs can lead to different test results.   \sqed
\end{example}

\begin{example}[Achievement awards; \citealt{angristlavy2009}]\label{ex:ag09} In this example, I reanalyze data from a randomized trial of \citet{angristlavy2009} in Israel. Their intervention provided cash rewards to low-achieving high school students if they performed well on the Bagrut certification exams  
for university admission in Israel. I follow the analysis in Table~5 of \citet{angristlavy2009} and focus on 32 schools in the sample for which Bagrut rates from 2000 to 2002 are available. Of these schools, 15 received treatment and 17 did not. Because 5 schools did not comply with treatment, the estimates below should be interpreted as intent-to-treat effects. Following \citeauthor{angristlavy2009}, I investigate the performance of girls in the June 2001 exams who were close to achieving Bagrut certification in the sense that they were ranked above the median of the credit-weighted January 2001 scores of girls. The sample also includes all girls who were above the median in 2000 and 2002. The 2948 girls who met these criteria had an above 50\% chance of Bagrut certification. I view each school over time as a cluster, which yields an average cluster size of approximately 92 students. 

\citet{angristlavy2009} report a large number of specifications. I consider a version of their fixed-effects model and estimate $Y_{i,t,k} = \theta_0 I_t + \delta D_k I_t + \eta J_t + \beta \mathit{top}_i + \zeta_k  + U_{i,t,k},$ where $i$ indexes students, $t$ indexes time, $k$ indexes schools, $Y_{i,k}$ indicates Bagrut status, $D_k$ is the treatment indicator, $I_t$ equals $1$ in 2001 and is $0$ otherwise, $J_t$ equals $1$ in 2002 and is $0$ otherwise, $\mathit{top}_i$ indicates whether a student is in the top quartile of the pre-Bagrut grade distribution of girls in the cohort, and $\zeta_k$ is a school fixed effect. \citeauthor{angristlavy2009} estimate several related specifications by logit in their Table~5. They report heteroskedasticity-robust standard errors for that table and argue that clustering is accounted for by their fixed effects.  For simplicity and ease of interpretation, I estimate the model by least squares. The model predicts an average increase in the probability of receiving Bagrut status by $0.114$ relative to a mean of $0.539$ with a robust standard error of $0.037$. A null of no effect against the alternative that $\delta$ is positive is rejected at any conventional significance level if standard normal critical values are used. This is in line with Table~5, col.~(3) of \citet{angristlavy2009}, who report significant effects ranging from $0.093$ to $0.168$ with standard errors ranging from $0.039$ to $0.045$ for this sample and several subsamples.

\begin{figure}[t]
\centering
\includegraphics[width=.9\textwidth]{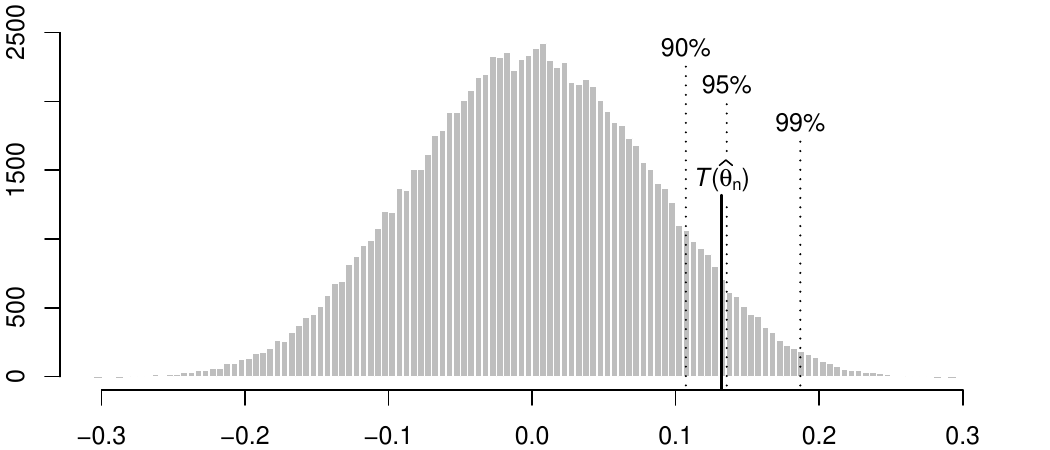}
\caption{Histogram of the permutation distribution of $T(\hat{\theta}_n) \approx 0.132$ (solid black line) from Example \ref{ex:ag09} with 90\%, 95\%, and 99\% critical values (dotted lines).}
\label{fig:ag09}
\end{figure}

To apply the adjusted permutation test, I view each cluster as an individual regression and separately estimate each of the $q=32$ equations in
\begin{equation*}
Y_{i,t,k} = 
\begin{cases} 
\theta_1 I_t + \eta J_t + \beta \mathit{top}_i + \zeta_k + U_{i,t,k}, &1\leq k\leq 15, \\ 
\theta_0 I_t + \eta J_t + \beta \mathit{top}_i + \zeta_k + U_{i,t,k}, &15< k\leq 32.
\end{cases}
\end{equation*}
Note that $\zeta_k$ is now simply the constant term in each regression. The resulting test statistic $T(\hat{\theta}_n) \approx 0.132$ can be viewed as an alternative point estimate of $\delta$ and is comparable in magnitude to the estimates reported in \citet{angristlavy2009}. However, as can be seen in Figure \ref{fig:ag09}, which plots the permutation distribution from 100,000 draws together with the corresponding critical values, the adjusted permutation test only rejects the null of no effect in favor of a positive effect at the 10\% level and barely does not reject at the 5\% level. If the fixed effects in the regression do not fully account for the within-cluster dependence in the data, the positive effect for girls may therefore be far less significant than previously reported. This result in also line with \citeauthor{angristlavy2009}, who find substantial but statistically marginal positive effects for girls across a wide variety of plausible specifications when they use cluster-robust standard errors. 
Also note that the 5\% and 10\% level one-sided tests performed here are outside the feasible range of the \citet{ibragimovmueller2016} test. For the \citet{canayetal2014} test, there are $17!/2 \approx 1.78\times 10^{14}$ ways of testing if 15 treated clusters are paired with 15 control clusters and two control clusters are dropped. In 1,000 randomly chosen unique pairings, the \citet{canayetal2014} test rejected the null of no effect against $\delta > 0$ for 425 pairings at the 5\% level and in 48 pairings at the 1\% level. Any desired conclusion could be reached by choosing a specific pairing.
\sqed

\end{example}

\bibliographystyle{chicago-ff}
\bibliography{qspec.bib}

\begin{thebibliography}{}

\bibitem[\protect\citeauthoryear{Bloom, Eifert, Mahajan, McKenzie, and
  Roberts}{Bloom et~al.}{2013}]{bloometal2013}
Bloom, N., B.~Eifert, A.~Mahajan, D.~McKenzie, and J.~Roberts (2013).
\newblock Does management matter? {E}vidence from {I}ndia.
\newblock {\em The Quarterly Journal of Economics\/}~{\em 128:1}, 1--51.

\bibitem[\protect\citeauthoryear{Davidson and MacKinnon}{Davidson and
  MacKinnon}{2000}]{davidsonmackinnon2000}
Davidson, R. and J.~G. MacKinnon (2000).
\newblock Bootstrap tests: how many bootstraps?
\newblock {\em Econometric Reviews\/}~{\em 19:1}, 55--68.

\bibitem[\protect\citeauthoryear{Keim}{Keim}{1983}]{keim1983}
Keim, D.~B. (1983).
\newblock Size related anomalies and stock return seasonality: further
  empirical evidence.
\newblock {\em Journal of Financial Economics\/}~{\em 12}, 13--32.

\bibitem[\protect\citeauthoryear{Rosenbaum}{Rosenbaum}{1984}]{rosenbaum1984}
Rosenbaum, P.~R. (1984).
\newblock Conditional permutation tests and the propensity score in
  observational studies.
\newblock {\em Journal of the American Statistical Association\/}~{\em 79},
  565--574.

\bibitem[\protect\citeauthoryear{S\l{}oczy\'nski}{S\l{}oczy\'nski}{2018}]{sloczynski2018}
S\l{}oczy\'nski, T. (2018).
\newblock A general weighted average representation of the ordinary and
  two-stage least squares estimands.
\newblock Working paper, Department of Economics, Brandeis University.

\end{thebibliography}


\begin{thebibliography}{}

\bibitem[\protect\citeauthoryear{Angrist and Lavy}{Angrist and
  Lavy}{2009}]{angristlavy2009}
Angrist, Joshua and Victor Lavy (2009).
\newblock The effects of high stakes high school achievement awards: Evidence
  from a randomized trial.
\newblock {\em American Economic Review\/}~{\em 99:4}, 301--331.

\bibitem[\protect\citeauthoryear{Bertrand, Duflo, and Mullainathan}{Bertrand
  et~al.}{2004}]{bertrandetal2004}
Bertrand, Marianne, Esther Duflo, and Sendhil Mullainathan (2004).
\newblock How much should we trust differences-in-differences estimates?
\newblock {\em Quarterly Journal of Economics\/}~{\em 119:1}, 249--275.

\bibitem[\protect\citeauthoryear{Bester, Conley, and Hansen}{Bester
  et~al.}{2011}]{besteretal2014}
Bester, C.~Alan, Timothy~G. Conley, and Christian~B. Hansen (2011).
\newblock Inference with dependent data using cluster covariance estimators.
\newblock {\em Journal of Econometrics\/}~{\em 165:2}, 137--151.

\bibitem[\protect\citeauthoryear{Cameron, Gelbach, and Miller}{Cameron
  et~al.}{2008}]{cameronetal2008}
Cameron, A.~Colin, Jonah~B. Gelbach, and Douglas~L. Miller (2008).
\newblock Bootstrap-based improvements for inference with clustered errors.
\newblock {\em Review of Economics and Statistics\/}~{\em 90:3}, 414--427.

\bibitem[\protect\citeauthoryear{Canay, Romano, and Shaikh}{Canay
  et~al.}{2017}]{canayetal2014}
Canay, Ivan~A., Joseph~P. Romano, and Azeem~M. Shaikh (2017).
\newblock Randomization tests under an approximate symmetry assumption.
\newblock {\em Econometrica\/}~{\em 85:3}, 1013--1030.

\bibitem[\protect\citeauthoryear{Canay, Santos, and Shaikh}{Canay
  et~al.}{2021}]{canayetal2018}
Canay, Ivan~A., Andres Santos, and Azeem~M. Shaikh (2021).
\newblock The wild bootstrap with a ``small'' number of ``large'' clusters.
\newblock {\em Review of Economics and Statistics\/}~{\em 103:2}, 346--363.

\bibitem[\protect\citeauthoryear{Conley and Taber}{Conley and
  Taber}{2011}]{conleytaber2011}
Conley, Timothy~G. and Christopher~R. Taber (2011).
\newblock Inference with ``difference in differences'' with a small number of
  policy changes.
\newblock {\em Review of Economics and Statistics\/}~{\em 93:1}, 113--125.

\bibitem[\protect\citeauthoryear{de~Chaisemartin and
  D'Haultf{\oe}ille}{de~Chaisemartin and
  D'Haultf{\oe}ille}{2020}]{chaisemartindhaultfoeille2020}
de~Chaisemartin, Clement and Xavier D'Haultf{\oe}ille (2020).
\newblock Two-way fixed effects estimators with heterogeneous treatment
  effects.
\newblock {\em American Economic Review\/}~{\em 110:9}, 2964--2996.

\bibitem[\protect\citeauthoryear{Djogbenou, MacKinnon, and Nielsen}{Djogbenou
  et~al.}{2019}]{djogbenouetal2019}
Djogbenou, Antoine, James~G. MacKinnon, and Morten Nielsen (2019).
\newblock Asymptotic theory and wild bootstrap inference with clustered errors.
\newblock {\em Journal of Econometrics\/}~{\em 212:2}, 393--412.

\bibitem[\protect\citeauthoryear{Donald and Lang}{Donald and
  Lang}{2007}]{donaldlang2007}
Donald, Stephen~G. and Kevin Lang (2007).
\newblock Inference with difference-in-differences and other panel data.
\newblock {\em Review of Economics and Statistics\/}~{\em 89:2}, 221--233.

\bibitem[\protect\citeauthoryear{{El Machkouri}, Voln{\'y}, and Wu}{{El
  Machkouri} et~al.}{2013}]{machkouriaetal2013}
{El Machkouri}, Mohamed, Dalibor Voln{\'y}, and Wei~Biao Wu (2013).
\newblock A central limit theorem for stationary random fields.
\newblock {\em Stochastic Processes and their Applications\/}~{\em 123:1},
  1--14.

\bibitem[\protect\citeauthoryear{Fama and MacBeth}{Fama and
  MacBeth}{1973}]{famamacbeth1973}
Fama, Eugene~F. and James~D. MacBeth (1973).
\newblock Risk, return, and equilibrium: Empirical tests.
\newblock {\em Journal of Political Economy\/}~{\em 81:3}, 607--636.

\bibitem[\protect\citeauthoryear{Gneiting}{Gneiting}{1997}]{gneiting1997}
Gneiting, Tilman (1997).
\newblock Normal scale mixtures and dual probability densities.
\newblock {\em Journal of Statistical Computation and Simulation\/}~{\em 59:4},
  375--384.

\bibitem[\protect\citeauthoryear{Hagemann}{Hagemann}{2019}]{hagemann2019}
Hagemann, Andreas (2019).
\newblock Placebo inference on treatment effects when the number of clusters is
  small.
\newblock {\em Journal of Econometrics\/}~{\em 213:1}, 190--209.

\bibitem[\protect\citeauthoryear{Hoeffding}{Hoeffding}{1952}]{hoeffding1952}
Hoeffding, Wassily (1952).
\newblock The large-sample power of tests based on permutations of
  observations.
\newblock {\em Annals of Mathematical Statistics\/}~{\em 23:2}, 169--192.

\bibitem[\protect\citeauthoryear{Ibragimov and M{\"u}ller}{Ibragimov and
  M{\"u}ller}{2010}]{ibragimovmueller2010}
Ibragimov, Rustam and Ulrich M{\"u}ller (2010).
\newblock {\itshape t}-statistic based correlation and heterogeneity robust
  inference.
\newblock {\em Journal of Business \& Economic Statistics\/}~{\em 28:4},
  453--468.

\bibitem[\protect\citeauthoryear{Ibragimov and M{\"u}ller}{Ibragimov and
  M{\"u}ller}{2016}]{ibragimovmueller2016}
Ibragimov, Rustam and Ulrich M{\"u}ller (2016).
\newblock Inference with few heterogenous clusters.
\newblock {\em Review of Economics and Statistics\/}~{\em 98:1}, 83--06.

\bibitem[\protect\citeauthoryear{Jenish and Prucha}{Jenish and
  Prucha}{2009}]{jenischprucha2009}
Jenish, Nazgul and Ingmar~R. Prucha (2009).
\newblock Central limit theorems and uniform laws of large numbers for arrays
  of random fields.
\newblock {\em Journal of Econometrics\/}~{\em 150:1}, 86–98.

\bibitem[\protect\citeauthoryear{MacKinnon and Webb}{MacKinnon and
  Webb}{2017}]{mackinnonwebb2014}
MacKinnon, James~G. and Matthew~D. Webb (2017).
\newblock Wild bootstrap inference for wildly different cluster sizes.
\newblock {\em Journal of Applied Econometrics\/}~{\em 32:2}, 233--254.

\bibitem[\protect\citeauthoryear{Roth, Sant’Anna, Bilinski, and Poe}{Roth
  et~al.}{2022}]{rothetal2022}
Roth, Jonathan, Pedro Sant’Anna, Alissa Bilinski, and John Poe (2022).
\newblock What’s trending in difference-in-differences? {A} synthesis of the
  recent econometrics literature.
\newblock Working paper, \texttt{arXiv:2201.01194}.

\bibitem[\protect\citeauthoryear{Sz\'ekely}{Sz\'ekely}{2006}]{szekely2006}
Sz\'ekely, Gabor~J. (2006).
\newblock Students {\emph{t}}-test for scale mixture errors.
\newblock {\em IMS Lecture Notes -- Monograph Series\/}~{\em 49}, 9--15.

\bibitem[\protect\citeauthoryear{White}{White}{2001}]{white2001}
White, Halbert (2001).
\newblock {\em Asymptotic Theory for Econometricians\/} (revised ed.).
\newblock Academic Press, San Diego.

\end{thebibliography}


\clearpage

\appendix
\setcounter{page}{1}
\renewcommand{\thepage}{\arabic{page} (Online Appendix)}
\renewcommand*{\thefootnote}{\fnsymbol{footnote}}
\linespread{1.3}\selectfont
\section*{\bfseries\MakeUppercase{Online Supplemental appendix to\\ ``Permutation inference with a finite\\ number of heterogeneous clusters''}\protect\footnote{Andreas Hagemann, University of Michigan.}}
\vskip 1em
\thispagestyle{empty}

This supplemental appendix is organized as follows: Appendix~\ref{s:moretheory} presents additional theoretical results, some of which are of potentially independent interest. Appendix~\ref{s:moreexamples} provides a step-by-step procedure for implementing the adjusted permutation test and applies that procedure in several examples. Appendix~\ref{s:morenumerical} contains additional numerical results and comparisons with the test of \citet{ibragimovmueller2016}. Appendix~\ref{s:proofs} contains proofs. Appendix~\ref{s:balpha} presents a simple algorithm for simulating critical values beyond those found in Table \ref{tb:balphavals} in the main text.

\section{Additional theoretical results}\label{s:moretheory}
I start with a discussion of the behavior of the test under the alternative $H_1\colon \mu_1 > \mu_0$. (Tests in the other direction follow by considering $-X$ instead of $X$.) Let $\delta = \mu_1-\mu_0$ and denote by $\Phi$ the standard normal distribution function. The distribution function  of $\max_{1\leq k\leq q_0}X_{q_1 + k}$ is equal to $x\mapsto \prod_{1\leq k\leq q_0}\Phi(x/\sigma_{k+q_1})=: \tilde{\Phi}(x)$ and therefore has a continuous and strictly increasing inverse. The following result gives a simple lower bound on the power of a permutation test as a function of $\delta$, $\Phi$, $\tilde{\Phi}$, and the standard deviations in the treatment group. Here I assume that the $\alpha$ under consideration is feasible, i.e., the corresponding $\balpha$ satisfies $\lceil (1-\balpha)|\Sym|\rceil < |\Sym|$ or, equivalently, $\balpha \geq 1/|\Sym|$. Otherwise the test becomes trivial because the null is never rejected.
\begin{theorem}[Power]\label{t:bfpower}
Suppose $X = (X_1,\dots, X_q)$ with independent $X_k\sim N(\mu + \delta 1\{k \leq q_1\}, \sigma^2_k)$,  $1\leq k\leq q$. Let $\balpha \geq 1/|\Sym|$. Then, for every $\sigma_1,\dots,\sigma_q>0$, \[\inf_{\mu\in\mathbb{R}}\prob \bigl(T(X) > T^{\balpha}(X,\Sym)\bigr)  \geq \int_0^1 \prod_{1 \leq j \leq  q_1} \Phi \Biggl( \frac{\delta - \tilde{\Phi}^{-1}(t)}{\sigma_j} \Biggr)dt.\]
\end{theorem}

As can be expected, the power of the test is driven by the strength of the signal $\delta$ relative to the noise represented by the standard deviations $\sigma_1,\dots, \sigma_q$. For example, a small treatment effect $\delta$ can be drowned out by large variation in the control group because $t\mapsto \tilde{\Phi}^{-1}(t)$ will then be positive and large for most values of $t$. However, the power of the test is not inherently limited. The integrand on the right is bounded by $1$ and converges to $1$ as $\delta \to \infty$ pointwise for every $t$. The integral and consequently the power of the permutation test therefore approach $1$ by dominated convergence as $\delta \to \infty$. Both the bound and this result can be generalized to the symmetric scale mixtures from Corollary \ref{t:symscalebound}; see Lemma \ref{l:symscalepower} for details.

Next, I discuss several aspects of the practical implementation of the permutation test \eqref{eq:adjustedtestdec}. First, one can still perform an asymptotic $\alpha$-level test if the observed data or statistic $X_n$ converges in distribution to the $X$ considered in Theorem~\ref{t:bfbound} or Corollary~\ref{t:symscalebound}. The reason is that the $g$ that order $T(gX_n)$ and $T(gX)$ as $g$ varies over $\Sym$ eventually coincide if sufficiently many entries of $X$ are smooth. The proof is a consequence of arguments in \citet{canayetal2014}. 
\begin{proposition}[Large sample approximation]\label{p:controlasy}
Let $X_n \leadsto X\in\mathbb{R}^q$ and let $T$ be as in \eqref{eq:compmean}. If $X$ has independent entries of which more than $q_1\wedge q_0$ are continuously distributed, then \[\lim_{n\to\infty} \prob \bigl( T(X_n) >  T^{(j)}(X_n,\Sym)\bigr) = \prob \bigl( T(X) >  T^{(j)}(X,\Sym)\bigr), \qquad \text{every~} 1\leq j\leq |\Sym|.\]
\end{proposition} 

Second, if evaluating $T(gX)$ over all elements of $\Sym$ is too costly because $|\Sym| = {q \choose q_1}$ is large, the computational burden can be reduced by working with a random sample $\Sym_m$ of $m$ draws from the uniform distribution on $\Sym$. This is often referred to as ``stochastic approximation.'' The following result shows that the critical values $T^{p}(X,\Sym_m)$ and $T^{p}(X,\Sym)$ lead to identical test decisions for any $p$ and large $m$ as long as $p|\Sym|$ is not an integer. If $p|\Sym|$ \emph{is} in fact an integer, the stochastic approximation can be marginally more conservative. The reason is that $p \mapsto T^{p}(X,\Sym)$ can vary discontinuously at integer values of $p|\Sym|$. The stochastic approximation then hits the order statistic just above $T^{p}(X,\Sym)$ with nonzero probability. The same arguments apply if the identity transformation is always included in $\Sym_m$, which is common practice for randomization tests.
\begin{proposition}[Stochastic approximation]\label{p:stochapprox}
Let $X_n\in\mathbb{R}^q$ be an arbitrary random vector possibly depending on $n$. Suppose $\Sym_m$ is a collection of $m$ random draws from $\Sym$ independent of $X_n$. Then \[\lim_{m\to\infty} \prob \bigl( T(X_n) >  T^{p}(X_n,\Sym_m)\bigr) \leq \prob \bigl( T(X_n) >  T^{p}(X_n,\Sym)\bigr), \qquad \text{every~} p\in (0,1),\] with equality unless $p|\Sym|\in\mathbb{N}$. The result remains true if one of the members of $\Sym_m$ is replaced by the identity with probability one.
\end{proposition}

As a referee points out, the choice of $m$ is important in practice. In particular, it seems if $|\Sym|$ is large, then $m$ must be large as well to provide an accurate stochastic approximation of the test decision. However, this is only true if the $p$-value $\hat{p}(X, \Sym_m)$, as defined in \eqref{eq:pval}, is very close to $\balpha$. If $\hat{p}(X, \Sym_m)$ is much larger than $\balpha$ for a given $m$, there is often enough information to conclude that $\hat{p}(X, \Sym)$ is highly unlikely to be smaller than $\balpha$. The same is true if the direction of the inequalities is reversed. The reason is that $\ev(\hat{p}(X, \Sym_m)\mid X) = \hat{p}(X, \Sym)$ and, for almost every realization of $X$, the central limit theorem implies that $\sqrt{m}(\hat{p}(X, \Sym_m) - \hat{p}(X, \Sym))$ converges to mean-zero normal with variance $\hat{p}(X, \Sym)(1-\hat{p}(X, \Sym))$. It is therefore easy to test hypotheses of the form $\hat{p}(X, \Sym) \geq \balpha$ or $\hat{p}(X, \Sym) \leq \balpha$ with a very small error tolerance $\beta$. For example, if $\hat{p}(X, \Sym_{m}) > \balpha$ for a given $m$, one can check whether $\hat{p}(X, \Sym) \leq \balpha$ can be rejected at this $m$. If not, one can add draws from $\Sym$ until the decision becomes possible. This idea is, in fact, the basis for the widely-used algorithm of \citetAppendix{davidsonmackinnon2000} for determining a sufficient number of bootstrap repetitions in models where the bootstrap is expensive to compute. Their algorithm can be adapted to the present problem with only notational changes.

\begin{algorithm}[Choosing $m$ if $|\Sym|$ is very large] Choose a starting value $m$ (e.g., 10,000), a step size $m'$ (e.g., 1,000), a maximal number of permutations $m_{\max}$ (e.g., 100,000), and an error tolerance $\beta$ (e.g., $.001$).
\begin{enumerate}
\item If  $\hat{p}(X, \Sym_{m}) < \balpha$, test the null hypothesis $\hat{p}(X, \Sym) \geq \balpha$ by rejecting in favor of $\hat{p}(X, \Sym) < \balpha$ if $\sqrt{m}(\hat{p}(X, \Sym_{m}) - \balpha)/\sqrt{\balpha(1-\balpha)} < \Phi^{-1}(\beta)$. Stop if the null is rejected and use $\hat{p}(X, \Sym_{m})$ as if it were $\hat{p}(X, \Sym)$.
\item If  $\hat{p}(X, \Sym_{m}) > \balpha$, test the null hypothesis $\hat{p}(X, \Sym) \leq \balpha$ by rejecting in favor of $\hat{p}(X, \Sym) > \balpha$ if $\sqrt{m}(\hat{p}(X, \Sym_{m}) - \balpha)/\sqrt{\balpha(1-\balpha)} > \Phi^{-1}(1-\beta)$. Stop if the null is rejected and use $\hat{p}(X, \Sym_{m})$ as if it were $\hat{p}(X, \Sym)$.
\item Stop if $m + m' > m_{\max}$ and use $\hat{p}(X, \Sym_{m})$ as if it were $\hat{p}(X, \Sym)$. Otherwise draw $m'$ additional permutations from $\Sym$, set $m = m + m'$, and restart from step (1). 
\end{enumerate}
\end{algorithm}

Finally, the two approximation results in Propositions~\ref{p:controlasy} and \ref{p:stochapprox} can be combined with Theorem~\ref{t:bfbound} to obtain \[ \lim_{n\to\infty} \lim_{m\to\infty} \prob \bigl( T(X_n) >  T^{\balpha}(X_n,\Sym_m)\bigr) \leq \alpha, \] i.e., adjusted permutation inference with an asymptotically normally distributed vector with heterogeneous variances remains approximately valid even if the set of permutations is drawn at random. It should also be noted that Proposition \ref{p:stochapprox} is generic and can be restated for other statistics $T$ and finite groups with appropriate notational changes. Proposition \ref{p:controlasy} can be extended to other statistics and groups under smoothness conditions.

\section{Additional examples}\label{s:moreexamples}
\label{s:examples}
I first present a brief summary of how the permutation test can be implemented in practice. By Theorem \ref{t:behrensfisherasy}, the following procedure provides an asymptotically $\alpha$-level test in the presence of a finite number of large clusters that are arbitrarily heterogeneous. The test is free of nuisance parameters, does not require matching clusters or any other decisions on part of the researcher, can be two-sided or one-sided in either direction, and is able to detect all fixed and $1/\sqrt{n}$-local alternatives.

\begin{algorithm}[Permutation test adjusted for cluster heterogeneity]\item[]
\begin{enumerate}
	\item Order the data such that clusters $1\leq k\leq q_1$ received treatment and clusters $q_1+1 \leq k\leq q_1 + q_0 = q$ did not. Compute for each $k = 1,\dots, q$ and using only data from cluster $k$ an estimate $\hat{\theta}_{n,k}$ of either $\theta_1$ or $\theta_0$ depending on whether $k$ received treatment or not so that the difference $\theta_1-\theta_0$ is the treatment effect of interest. (Examples are provided below and in the main text.) Define $\hat{\theta}_n = (\hat{\theta}_{n,1},\dots, \hat{\theta}_{n,q})$ and compute $T(\hat{\theta}_n) = q_1^{-1}\sum_{k=1}^{q_1} \hat{\theta}_{n,k} - q_0^{-1} \sum_{k=q_1 + 1}^q \hat{\theta}_{n,k}$.
	\item For the desired $\alpha$, choose $\balpha$ from Table \ref{tb:balphavals}.
	\item Compute the set of permutations $\Sym$ defined in \eqref{eq:combinations}. Alternatively, draw a large random sample of permutations $\Sym_m$ and replace $\Sym$ by $\Sym_m$ in step \eqref{al:practice4}.
	\item\label{al:practice4} Reject the null hypothesis of no effect of treatment $H_0\colon \theta_1 = \theta_0$ against
	\begin{enumerate}
		\item $\theta_1 > \theta_0$ if $T(\hat{\theta}_n) >  T^{\balpha}(\hat{\theta}_n,\Sym)$ for a test with asymptotic level $\alpha$,
		\item $\theta_1 < \theta_0$ if $T(-\hat{\theta}_n) >  T^{\balpha}(-\hat{\theta}_n,\Sym)$ for a test with asymptotic level $\alpha$,
		\item $\theta_1 \neq \theta_0$ if $T(\hat{\theta}_n) >  T^{\balpha}(\hat{\theta}_n,\Sym)$ or $T(-\hat{\theta}_n) >  T^{\balpha}(-\hat{\theta}_n,\Sym)$ for a test with asymptotic level $2\alpha$,
	\end{enumerate}
	where $T^{\balpha}(\cdot, \G)$, defined in \eqref{eq:critval}, is the $\lceil (1-\balpha)|\G|\rceil$-th largest value of the permutation distribution of $T(\cdot)$.
\end{enumerate}
\end{algorithm}

I now discuss two additional examples of how the cluster-level statistics $\hat{\theta}_n$ can be constructed such that the condition \eqref{eq:jointconv} required for Theorem \ref{t:behrensfisherasy} holds. For simplicity, the discussion focuses on \eqref{eq:jointconv} under the null hypothesis $H_0 : \theta_1 = \theta_0$ but the arguments apply more broadly.
\begin{example}[Regression with cluster-level treatment]\label{ex:clusterreg}
Consider a linear regression model
\begin{equation*}
Y_{i,k} = \theta_0 + \delta D_{k} + \beta_k' X_{i,k} + U_{i,k},
\end{equation*}
where $i$ indexes individuals within clusters $1\leq k\leq q$. The parameter of interest is the coefficient $\delta$ on the treatment dummy $D_k$ indicating whether cluster $k$ received treatment or not. The regression also includes covariates $X_{i,k}$ that vary within each cluster and have coefficients $\beta_k$ that may vary across clusters. The condition $\ev(U_{i,k}\mid D_k, X_{i,k}) = 0$ identifies $\theta_1 = \theta_0 + \delta$ within a treated cluster and $\theta_0$ within an untreated cluster. The preceding display can then be written as
\[ 
Y_{i,k} = 
\begin{cases}
\theta_1 + \beta_k' X_{i,k} + U_{i,k}, &1\leq k\leq q_1,\\ 
\theta_0 + \beta_k' X_{i,k} + U_{i,k}, &q_1< k\leq q.
\end{cases}
\]
View these as $q$ separate regressions and use the least squares estimates of the constants $\theta_1$ and $\theta_0$ as $\hat{\theta}_n = (\hat{\theta}_{n,1},\dots, \hat{\theta}_{n,q})$. Also note that permuting $\hat{\theta}_n$ is identical to permuting the vector of the observed treatment indicators that labels each of these $q$ regressions as coming from either a treated or an untreated cluster. The same types of arguments as in Example \ref{ex:diffindiff} can be used to establish a central limit theorem for $\hat{\theta}_n$.

Under suitable conditions, the $\delta$ in this example can be interpreted as an average treatment effect in a potential outcomes framework. See, e.g., \citetAppendix{sloczynski2018} and references therein for a precise discussion. The goal here is to make permutation inference about $\delta$. This should not be confused with testing the ``sharp'' null hypothesis that the treatment and control potential outcomes under the intervention are identical. Testing sharp nulls is often associated with permutation testing and is a much stronger restriction than that the \emph{average} effect $\delta$ on the outcomes be zero. \citetAppendix{rosenbaum1984} explains how to use permutation inference to test sharp nulls in the presence of covariates under assumptions on the propensity score. \sqed
\end{example}

\begin{example}[Binary choice with cluster-level treatment]\label{ex:binreg} Consider a version of the model in Example \ref{ex:clusterreg} as the latent model $Y_{i,k} = \theta_0 + \delta D_{k} + \beta_0' X_{i,k} + U_{i,k}$ in a binary choice setting. Here $U_{i,k}$ has a known, smooth, and symmetric distribution function $F$ and is independent of $(D_{k},X_{i,k})$. Only $1\{Y_{i,k} > 0\}$, $X_{i,k}$, and $D_k$ are observed. Each cluster has $n_k$ observations and can be viewed as a separate binary choice model 
\[ 
P(Y_{i,k} > 0 \mid  X_{i,k}) = 
\begin{cases}
F(\theta_1 + \beta_0' X_{i,k}), &1\leq k\leq q_1,\\ 
F(\theta_0 + \beta_0' X_{i,k}), &q_1< k\leq q.
\end{cases}
\]
 If the treatment effect of interest is $F(\theta_1 + \beta_0' x) - F(\theta_0 + \beta_0' x)$ for some $x$, then $H_0\colon \theta_1 = \theta_0$ corresponds to the null hypothesis of no treatment effect. Let $\psi_{\theta, \beta}(y, x) = (1, x')'(1\{ y > 0 \} - F(\theta + \beta'x) )$ and suppose the moment condition $\ev \psi_{\theta_0, \beta_0}(Y_{i,k}, X_{i,k}) = 0$ holds for every $i$ and $k$. The corresponding $Z$-estimates $(\hat{\theta}_{n,k}, \hat{\beta}_{n,k}')'$ for the $k$-th cluster are zeros of
$\Psi_{n, k}(\theta, \beta) = n_k^{-1}\sum_{i=1}^{n_k} \psi_{\theta, \beta}(Y_{i,k}, X_{i,k}).$
Denote the derivative of $\Psi_{n, k}$ with respect to $(\theta, \beta')$ by $\dot{\Psi}_{n,k}$.

Using the same limit theory as outlined in Example~\ref{ex:diffindiffcont}, it is possible to argue under regularity conditions that $\dot{\Psi}_{n,k}$ converges pointwise in probability to a limit $\dot{\Psi}_{k}$ and $(\hat{\theta}_{n,k}, \hat{\beta}_{n,k})\pto (\theta_0,\beta_0)$. If $\dot{\Psi}_k(\theta_0, \beta_0)$ is non-singular and $\sqrt{n} \Psi_{n, k}(\theta_0, \beta_0) = O_P(1)$, then
\[ \sqrt{n}(\hat{\theta}_{n,k} - \theta_0) = e_1' \dot{\Psi}_k(\theta_0, \beta_0)^{-1} \sqrt{n} \Psi_{n, k}(\theta_0, \beta_0) + o_P(1), \] where $e_1$ is a conformable vector with a $1$ in the first position and $0$ otherwise. Condition \eqref{eq:jointconv} is satisfied if a central limit theorem applies to $\sqrt{n} \Psi_{n, k}(\theta_0, \beta_0)$. Because this is a scaled average of mean-zero random vectors, the same references as in Example~\ref{ex:diffindiffcont} can be used to establish a central limit theorem.  \sqed
\end{example}

\section{Additional numerical results}\label{s:morenumerical}
This section presents a detailed comparison of the \citet{ibragimovmueller2016} and adjusted permutation tests in Monte Carlo experiments and empirical examples.
\begin{example}[Equality of means]\label{ex:nloc}
The adjusted permutation test developed here and the \citet{ibragimovmueller2016} test both rely on results about the behavior of heterogeneous normal variables applied to certain test statistics. For the adjusted permutation test, this statistic is the comparison on means $T$. For the \citetalias{ibragimovmueller2016} test, it is the studentized two-sample statistic \[ \frac{\bar{X}_1 - \bar{X}_0}{\sqrt{\frac{1}{q_1(q_1-1)}\sum_{k=1}^{q_1}(X_k - \bar{X}_1)^2 + \frac{1}{q_0(q_0-1)}\sum_{k=q_1 + 1}^{q}(X_k - \bar{X}_0)^2}}, \] where $\bar{X}_1 = q_1^{-1}\sum_{k=1}^{q_1}X_k$ and $\bar{X}_0 = q_0^{-1}\sum_{k=q_1 + 1}^{q}X_k$. This statistic is compared to the quantiles of the Student $t$ distribution with $(q_1\wedge q_0)-1$ degrees of freedom. This example investigates the relative performance of the two tests. 

As in Section \ref{s:behrensfisher}, suppose $X = (X_1,\dots, X_q)\in\mathbb{R}^q$ has independent entries $X_k = \mu_0 + (\mu_1 - \mu_0)1\{k\leq q_1\} + \sigma_k Z_k$ with $Z_k$ distributed as $N(0, 1)$. The results reported here use $\mu_0 = 0$. To investigate the impact of heterogeneity on the two tests, I considered the following six configurations of $\sigma_1,\dots,\sigma_q$:
\begin{enumerate}[(a)]
\item $\sigma_1, \dots, \sigma_q = 1$,
\item $\sigma_1, \dots,\sigma_{q-1} = 1$, $\sigma_q = 100$
\item $\sigma_1, \dots, \sigma_{q_1-1} = 1$, $\sigma_{q_1} = 100$, $\sigma_{q_1+1}, \dots, \sigma_{q-1} = 1$, $\sigma_q = 100$,
\item $\sigma_1, \dots, \sigma_{q_1} = 1$, $\sigma_{q_1+1}, \dots, \sigma_{q} = 3$
\item $\sigma_1, \dots, \sigma_{q_1/2} = 3$, $\sigma_{q_1/2+1}, \dots, \sigma_{q_1 + q_0/2} = 1$, $\sigma_{q_1 + q_0/2+1}, \dots, \sigma_{q} = 3$, 
\item $\sigma_1, \dots, \sigma_{q_1/2} = 1$, $\sigma_{q_1/2+1}, \dots, \sigma_{q_1 + q_0/2} = 3$, $\sigma_{q_1 + q_0/2+1}, \dots, \sigma_{q} = 9$.
\end{enumerate}
Configurations (a), (d), (e), and (f) are taken from \citet{ibragimovmueller2016}.

\begin{figure}[htp]
\centering
\includegraphics[width=.85\textwidth]{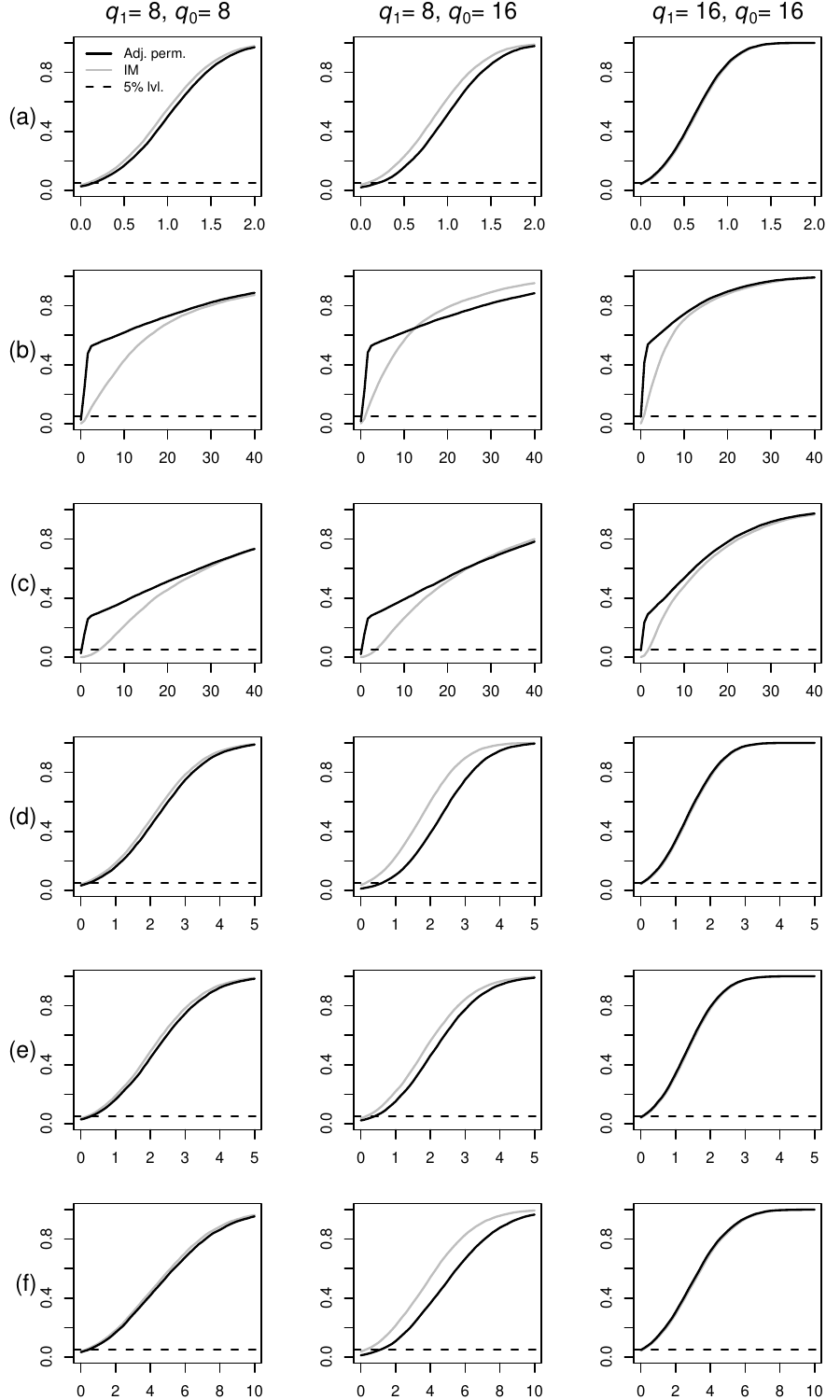}
\caption{Rejection frequencies of the adjusted permutation test (black lines) and the \citetalias{ibragimovmueller2016} test (IM, grey) for models (a)-(f) (rows) in Example~\ref{ex:nloc} for $q_1 = q_0 = 8$ (left), $q_1 = 8, q_0 = 16$ (middle), and $q_1 = q_0 = 16$ (right).}  \label{fig:norm_loc}
\end{figure}

Rows (a)-(f) of Figure~\ref{fig:norm_loc} correspond to the six configurations (a)-(f) and show the rejection frequencies of the adjusted permutation test (black lines) and the \citetalias{ibragimovmueller2016} test (grey) at the 5\% level (dashed line) as $\mu_1$ increases. The null hypothesis is correct at $\mu_1 = 0$. The columns correspond, from left to right, to the sample sizes $(q_1=8, q_0=8)$, $(q_1=8, q_0=16)$, and $(q_1=16, q_0=16)$. Each horizontal coordiate was computed from 10,000 Monte Carlo replications. As can be seen, the variation in $\sigma_k$ led to marked differences in power at different levels of heterogeneity. The adjusted permutation test was able to reject far more false nulls than the \citetalias{ibragimovmueller2016} test for small $\mu_1$ when there were few large variances as in (b) and (c). For instance, in (b) with $(q_1=8, q_0=8)$ at $\mu_1 = 1$ the adjusted permutation test rejected in 47.62\% of all cases whereas the \citetalias{ibragimovmueller2016} test rejected in only 6.36\% of all cases. This difference eventually disappeared for large $\mu_1$. However, neither test is more powerful. With slightly different variances within or across groups as in (d) and (f), the \citetalias{ibragimovmueller2016} test had an advantage when the sample sizes differed substantially. The differences between the two tests were much smaller for the other configurations. Other samples sizes (not shown) led to qualitatively similar results. 

\begin{figure}[htp]
\centering
\includegraphics[width=.85\textwidth]{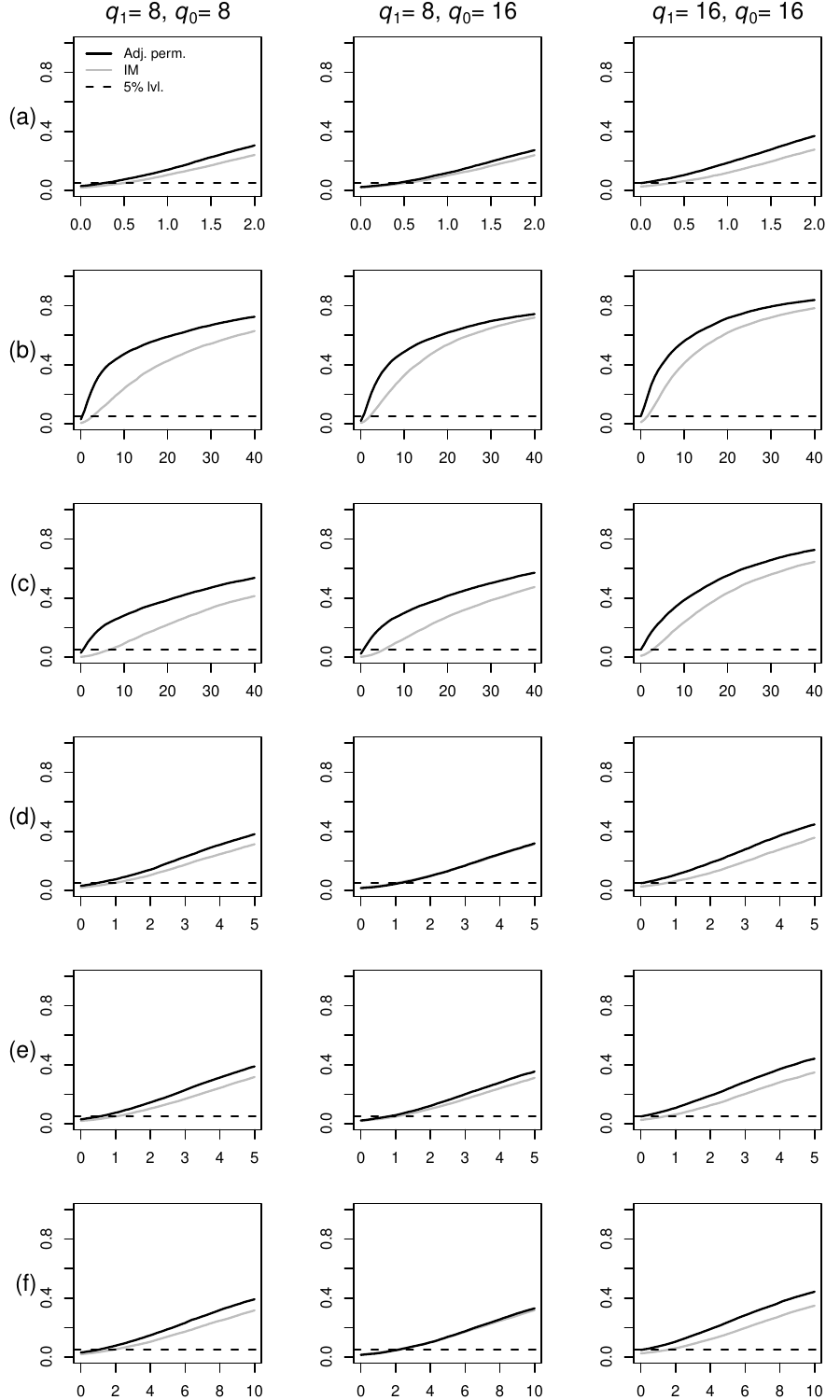}
\caption{Rejection frequencies of the adjusted permutation test (black lines) and the \citetalias{ibragimovmueller2016} test (IM, grey) as in Figure \ref{fig:norm_loc} but with Cauchy distributions.}  \label{fig:cauchy_loc}
\end{figure}

As a referee points out, it would be interesting to compare the performance of the adjusted permutation test and the \citetalias{ibragimovmueller2016} test in fat-tailed settings. Just like the adjusted permutation test, the \citetalias{ibragimovmueller2016} test can be used with mixtures of normals, which includes models with infinite variances. I therefore repeated the above experiments with standard Cauchy distributed $Z_k$ instead of standard normal distributions, holding all else equal. The results are plotted in Figure \ref{fig:cauchy_loc}. As can be seen, within the scope of the configurations for (a)-(f), the adjusted permutation test was more powerful than the \citetalias{ibragimovmueller2016} test for \emph{every} configuration at all sample sizes and for all values of $\mu_1$. In sharp contrast to the situation with standard normal $Z_k$, this was true even when the samples sizes differed.
\sqed 
\end{example}

A reviewer also recommends comparing the conclusions of adjusted permutation inference and the \citetalias{ibragimovmueller2016} test in empirical examples discussed in \citet{ibragimovmueller2016}, which include tests of hypotheses on January effects and a randomized trial of \citetAppendix{bloometal2013}. 
\begin{example}[January effects; \citealtAppendix{keim1983}]
\citetAppendix{keim1983} investigates January effects in stock returns. He considers excess returns in portfolios constructed from firms in the top and bottoms decile of size, as measured by market value of equity on the New York Stock Exchange (NYSE) and American Stock Exchange (now called NYSE American) over the period 1963-1979. To test whether the January effect is time invariant, \citeauthor{ibragimovmueller2016} assume that the data are suitably approximated by a scale mixture of normals and implement their test by comparing the January excess returns for 1963-1969 to the January excess returns for 1970-1979. They do not reject the null hypothesis of time invariance at the 5\% level but reject at the 10\% level. The adjusted permutation test does not reject at either significance level. \sqed
\end{example}

\begin{example}[Modern management practices; \citealtAppendix{bloometal2013}]
In this example, I reanalyze data form a randomized trial of \citetAppendix{bloometal2013}. Their intervention provided five months of extensive management consulting from a large international consulting firm to eleven randomly selected Indian textile plants. A control group of six randomly selected plants received only one month of diagnostic consulting. The experiment ran from 2008 to 2011 and several key performance measures were collected before, during, and after the intervention. These measures include data on quality defects, inventory, output, and total factor productivity. Here I focus on output because it is the only measure that has data for all 17 firms available. For the effect on output in their main results in their Table II, \citetAppendix{bloometal2013} run a regression of the log of picks (one pick is a single rotation of a weaving shuttle) on a treatment dummy, time fixed effects, and firm fixed effects. They find a 9\% increase in output as a result of the intervention. 

\citetAppendix{bloometal2013} use, among other methods, the \citet{ibragimovmueller2016} test to conduct inference. The adjusted permutation test also applies and can be computed as outlined in Examples \ref{ex:diffindiff} and \ref{ex:diffindiffcont}. Both the \citetalias{ibragimovmueller2016} and the adjusted permutation test find a significant positive effect on log output at the 5\% level, which confirms that the results of \citetAppendix{bloometal2013} remain valid even if methods designed for a small number of arbitrarily heterogeneous clusters are used.\sqed

\end{example}

\section{Proofs}\label{s:proofs}
\begin{proof}[Proof of Theorem \ref{t:bfbound} and Corollary \ref{t:symscalebound}]
 Denote the distribution function of an arbitrary random variable $Y$ by $F_Y$. We have $T(X) > T^{(|\Sym|-1)}(X,\Sym)$ if and only if $T(X) = T^{(|\Sym|)}(X,\Sym)$. Because the test statistic is location invariant, assume without loss of generality that $\mu = 0$. Denote by $X_{(1)}, X_{(2)}, \dots, X_{(q)}$ the order statistics of $X$. Then $T^{(|\Sym|)}(X,\Sym) = q_1^{-1}\sum_{k=1}^{q_1} X_{(k + q_0)} - q_0^{-1}\sum_{k=1}^{q_0} X_{(k)}$. Because $T(X) = T^{(|\Sym|)}(X,\Sym)$ and $\min\{X_1,\dots, X_{q_1}\} < \max\{X_{q_1+1}, \dots, X_q\}$ cannot be true at the same time and $\min\{X_1,\dots, X_{q_1}\} > \max\{X_{q_1+1}, \dots, X_q\}$ implies $T(X) = T^{(|\Sym|)}(X,\Sym)$, it follows that $\prob (T(X) = T^{(|\Sym|)}(X,\Sym))$ equals
\begin{align*}
&\prob\bigl(\min\{X_1,\dots, X_{q_1}\} > \max\{X_{q_1+1}, \dots, X_q\}\bigr)\\ &\qquad+ \prob\bigl(T(X) = T^{(|\Sym|)}(X,\Sym), \min\{X_1,\dots, X_{q_1}\} = \max\{X_{q_1+1}, \dots, X_q\}\bigr).
\end{align*}

Suppose $X_k = S_k Z_k$, $1\leq k\leq q$, where the $S_k$ is nonzero with probability one and the $Z_k$ has a continuous distribution. The second line of the preceding display must then be zero conditional on $S = (S_1\dots, S_q)$ and the same must therefore hold unconditionally.  The first line conditional on $S_1=\sigma_1, \dots, S_q=\sigma_q$ for fixed scales $\sigma_1,\dots, \sigma_q$ is, by independence, equivalent to the statement $\prob(\min\{X_1,\dots, X_{q_1}\} > \max\{X_{q_1+1}, \dots, X_q\}$) with $X_k = \sigma_k Z_k$ for $1\leq k\leq q$. In the following, I will therefore work with $X_k = \sigma_k Z_k$ first and return to the unconditional case later.	

Let $V = \max\{X_1,\dots, X_{q_1}\}$ and $W = \max\{X_{q_1+1}, \dots,$ $X_q\}$. Symmetry of $X_1,\dots, X_{q_1}$ and independence of $V$ and $W$ imply \[ \prob\bigl(\min\{X_1,\dots, X_{q_1}\} > W \bigr) = \prob\bigl(\min\{-X_1,\dots, -X_{q_1}\} > W\bigr) = \prob(V + W < 0 ).\] Suppose $q_1 < q_0$. The two maxima $V$ and $W$ must satisfy \[ \prob(V + W < 0 ) = \prob\Biggl(\bigcap_{k=1}^{q_1} \bigcap_{l=1}^{q_0} \{ X_k + X_{l + q_1} < 0 \} \Biggr) \leq  \prob\Biggl(\bigcap_{k=1}^{q_1} \{ X_k + X_{k+q_1} < 0 \} \Biggr). \] Define $Y_k = X_k + X_{k+q_1}$. Note that the $Y_k$ are independent across $1 \leq k\leq q_1$ and symmetric because $\prob(X_k + X_{k + q_1} \leq y) = \prob(-X_k -X_{k + q_1} \leq y) = \prob(-Y_k \leq y)$. The right-hand side of the preceding display then equals $\prob(\max\{Y_1,\dots, Y_{q_1}\} < 0) = F_Y(0)^{q_1}$. Conclude from symmetry that $\prob(V + W < 0 ) \leq 0.5^{q_1}$. Repeat the argument with $q_1 > q_0$ to obtain \[ \prob(V + W < 0 ) \leq \max\{0.5^{q_1}, 0.5^{q_0}\} = 2^{\min\{q_1,q_0 \}},\] as desired. To see that this bound is tight, assume first that $q_1 \geq q_0$. Choose $\sigma_1=\dots = \sigma_{q_1} = 1$,  $\sigma = \sigma_{q_1 + 1}= \dots = \sigma_{q}$, and let $U = \max\{ Z_{1+q_1}, \dots, Z_q \}$. Then $\prob(V+W < 0) = \ev F_V(-\sigma U)$. If $U > 0$, then $F_V(-\sigma U)\to 0$ almost surely as $\sigma \to\infty$. If $U < 0$, then $F_V(-\sigma U)\to 1$ almost surely as $\sigma \to\infty$. Conclude from dominated convergence that  $\ev F_V(-\sigma U) \to \prob(U < 0) = 0.5^{q_0}$. If $q_1 < q_0$, switch $V$ and $W$. This proves the theorem.

For the corollary, return to $X_k = S_k Z_k$ and redefine $V,W$ accordingly. It is still true that $\prob(V + W\mid S) \leq 1/2^{\min\{q_1,q_0\}}$ almost surely and therefore $\prob(T(X) > T^{(|\Sym|-1)}(X,\Sym)) \leq 1/2^{\min\{q_1,q_0\}}$, as required for the corollary.
\end{proof}

Define $V = \max\{ S_1 Z_1,\dots, S_{q_1}Z_{q_1} \}$ and $W = \max\{S_{q_1+1}Z_{q_1+1}, \dots,$ $S_q Z_q \}$. Let $w\mapsto F_W(w\mid S)$ be the distribution function of $W$ conditional on $S$.
\begin{lemma}\label{l:symscalepower}
	Suppose $X = (X_1,\dots, X_q)$ with $X_k = \mu + \delta 1\{k\leq q_1 \} + S_k Z_k$,  $1\leq k\leq q$, where the $Z_1,\dots, Z_q$ are iid copies of a random variable $Z$ with continuous distribution function and $S = (S_1, \dots, S_q$) is a random vector independent of $Z_1,\dots, Z_q$ with $P(S_k > 0) = 1$ for $1\leq k\leq q$. If $Z$ and $-Z$ have the same distribution, then  \[\inf_{\mu\in\mathbb{R}}\prob \bigl(T(X) > T^{\balpha}(X,\Sym)\bigr)  \geq \ev \int_0^1 \prod_{1 \leq j \leq  q_1} F_Z \Biggl( \frac{\delta - F^{-1}_W(t\mid S)}{S_j} \Biggr)dt.\] The right-hand side converges to $1$ as $\delta \to \infty$.
\end{lemma}

\begin{proof}[Proof of Lemma \ref{l:symscalepower}.]
This proof is similar to the proof of Theorem \ref{t:bfbound}. As before, consider $T(X) = T^{(|\Sym|)}(X,\Sym)$ and assume without loss of generality the case $\mu = 0$ so that $\min\{X_1,\dots,X_{q_0}\}$ has the same distribution as $\delta -V$. Because $T^{(|\Sym|)}(X,\Sym) = q_1^{-1}\sum_{k=1}^{q_1} X_{(k+q_0)} - q_0^{-1}\sum_{k=1}^{q_0} X_{(k)}$, continuity implies
\begin{align}\label{eq:symscalerep}
\prob \bigl(T(X) = T^{(|\Sym|)}(X,\Sym)\bigr)
= \prob(V + W < \delta).
\end{align}
Independence of $V$ and $W$ conditional on $S$ and continuity imply that there is an independent standard uniform $U$ such that the preceding display equals \[ \ev \prob\bigl( V < \delta - F_W^{-1}(U\mid S)\mid S\bigr) = \ev \int_0^1 \prob\bigl( V < \delta - F_W^{-1}(t\mid S) \mid S \bigr) dt, \] where the equality follows from Tonelli's theorem. By independence, distribution function of $V$ conditional on $S$ is $v\mapsto \prod_{1\leq j\leq q_1}F_Z(v/S_j)$. The first result now follows because $\prob (T(X) > T^{\balpha}(X,\Sym)) \geq \prob(T(X) = T^{(|\Sym|)}(X,\Sym)).$ The second result follows from \eqref{eq:symscalerep} as $\delta \to\infty$.
\end{proof}

\begin{proof}[Proof of Theorem \ref{t:bfpower}.]
This follows immediately from Lemma \ref{l:symscalepower} by letting $S = (\sigma_1,\dots, \sigma_q)$ with probability one and $F_Z = \Phi$.
\end{proof}

\begin{proof}[Proof of Proposition \ref{p:controlasy}]
Following \citet{canayetal2014}, I only have to show that for any two distinct $g, g'\in \Sym$, either $T(gx) = T(g'x)$ for all $x\in\mathbb{R}^q$ or $\prob(T(gX)\neq T(g'X)) = 1$. Let $w_{g(k)} = q_1^{-1} {1\{g(k) \leq q_1\}} - q_0^{-1} {1\{g(k) > q_1\}}$ and notice that $g \neq g'$ implies that $w_{g(k)}\neq w_{g'(k)}$ for at least two $k', k''\in\{1,\dots, q\}$. By the pigeonhole principle, $X_{k'}$ or $X_{k''}$ must be continuously distributed. Then $T(gX) - T(g'X) = \sum_{k=1}^q (w_{g(k)}- w_{g'(k)}) X_k$ is continuously distributed by independence and therefore $\prob(T(gX) - T(g'X) = 0) = 0$.
\end{proof}

\begin{proof}[Proof of Proposition \ref{p:stochapprox}] All limits are as $m\to\infty$. Let $\Sym_m = \{G_1,\dots, G_m\}$ be a collection of $m$ draws from the uniform distribution on $\Sym$, in which case $\ev(\hat{p}(X, \Sym_m)\mid X) = \hat{p}(X, \Sym)$. For almost every realization of $X$, the central limit theorem implies that $\sqrt{m}(\hat{p}(X, \Sym_m) - \hat{p}(X, \Sym))$ converges to mean-zero normal with variance $\hat{p}(X, \Sym)(1-\hat{p}(X, \Sym))$. Because $\hat{p}(X, \Sym)\geq 1/|\Sym|$, this variance can only be zero if $\hat{p}(X, \Sym) = 1$. This occurs if and only if $T(gX)\geq T(X)$ for all $g\in\Sym$, which also implies $\hat{p}(X, \Sym_m)=1$ for such $X$. 

By the equivalence of $p$-values and critical values, $T(X) >  T^{p}(X,\Sym_m)$ if and only if $\hat{p}(X, \Sym_m) \leq p$ and therefore \[ \prob \bigl( T(X) >  T^{p}(X,\Sym_m)| X\bigr) = \prob \Bigl( \sqrt{m}\bigl(\hat{p}(X, \Sym_m) - \hat{p}(X, \Sym)\bigr) \leq \sqrt{m}\bigl(p - \hat{p}(X, \Sym)\bigr)| X \Bigr). \] Since $\prob ( \sqrt{m}(p(X, \Sym_m) - p(X, \Sym)) \leq t\mid X )$ converges almost surely to a (possibly degenerate) normal distribution function, for every $\varepsilon > 0$ and almost every realization of $X$ there is an $M$ (possibly depending on $\varepsilon$ and $X$) such that the limit of $\prob ( \sqrt{m}(p(X, \Sym_m) - p(X, \Sym)) \leq -M\mid X )$ is at most $\varepsilon$ and $\prob ( \sqrt{m}(p(X, \Sym_m) - p(X, \Sym)) \leq M\mid X )$ is at least $1-\varepsilon$. If $p > p(X, \Sym)$, then $\sqrt{m}(p - p(X, \Sym))$ is eventually larger than every such $M$. If $p < p(X, \Sym)$, then $\sqrt{m}(p - p(X, \Sym))$ is eventually smaller than  $-M$. If $p = p(X, \Sym)$, which cannot occur if $p(X, \Sym) = 1$, the preceding display converges almost surely to 0.5. Conclude that the preceding display converges almost surely to  $1\{ p(X, \Sym) < p \}  + 1\{ p(X, \Sym) = p \}/2$. The dominated convergence theorem then implies \[\prob \bigl( T(X) >  T^{p}(X,\Sym_m)\bigr) \to \prob\bigl( \hat{p}(X, \Sym) < p \bigr) + 0.5\prob\bigl( \hat{p}(X, \Sym) = p \bigr). \] The right hand side is equal to $\prob( \hat{p}(X, \Sym) \leq p)$ if $\prob( \hat{p}(X, \Sym) = p) = 0$, which is the case if $p |\Sym|$ is not an integer because infinitesimal changes in $p$ cannot change $\prob( \hat{p}(X, \Sym) \leq p)$. If $\prob( \hat{p}(X, \Sym) = p)$ is nonzero, then the preceding display is smaller than $\prob( \hat{p}(X, \Sym) \leq p)$.

If $\Sym_m' = \{\id,G_2,\dots, G_m\}$ then, both unconditionally and conditional on $X$, \[\sqrt{m}\bigl(\hat{p}(X, \Sym_m') - \hat{p}(X, \Sym_m)\bigr) = \frac{1 - 1\{T(G_1 X)\geq T(X)\}}{\sqrt{m}} \pto 0.\] The proof now follows from the arguments for $\Sym_m$.
\end{proof}

\begin{proof}[Proof of Theorem \ref{t:behrensfisherasy}]
Suppose $\theta_1 = \theta_0$. Let $1_q$ denote a $q$-vector of ones and $X = (X_1,\dots, X_q)\sim N(0, \diag (\sigma^2_1,\dots, \sigma^2_q)(\theta_0))$. Notice that $T(\hat{\theta}_n) >  T^{\balpha}(\hat{\theta}_n,\Sym)$ if and only if \[T\bigl(\sqrt{n}(\hat{\theta}_n-\theta_0 1_q)\bigr) >  T^{\balpha}\bigr(\sqrt{n}(\hat{\theta}_n-\theta_0 1_q),\Sym\bigr).\] Hence, it suffices to prove the result with $X_n  = \sqrt{n}(\hat{\theta}_n-\theta_0 1_q)$ in place of $\hat{\theta}_n$. Because $X_n\leadsto X$, the desired result for $\theta_1 = \theta_0$ follows from Proposition \ref{p:controlasy} and Theorem \ref{t:bfbound}.

Suppose $\theta_1 = \theta_0 + \delta/\sqrt{n}$. Let $X_n = \sqrt{n}(\hat{\theta}_{n,k} - \theta_{1\{k\leq q_1\}})_{1\leq k\leq q}$ and $\Delta = (\delta1\{k\leq q_1\})_{1\leq k\leq q}$. Then $X_n + \Delta \leadsto X + \Delta$ by the assumed continuity and the Slutsky lemma. By construction, $T(\hat{\theta}_n) >  T^{\balpha}(\hat{\theta}_n,\Sym)$ is equivalent to $T(X_n + \Delta) > T^{\balpha}(X_n + \Delta,\Sym)$. Proposition \ref{p:controlasy} then implies \[ \prob \bigl( T(X_n+\Delta) >  T^{\balpha}(X_n + \Delta,\Sym)\bigr) \to \prob \bigl(T(X+\Delta) >  T^{\balpha}(X + \Delta,\Sym)\bigr). \] Now apply the lower bound developed in Theorem \ref{t:bfpower} to the right-hand side.

Suppose $\theta_1 = \theta_0 + \delta$. Let $\Delta_n = \sqrt{n}(\delta 1\{k\leq q_1\})_{1\leq k\leq q}$ so that $T(\hat{\theta}_n) \leq  T^{\balpha}(\hat{\theta}_n,\Sym)$ is equivalent to $T(X_n) \leq T^{\balpha}(X_n + \Delta_n,\Sym) - T(\Delta_n)$. For a large $M > 0$, the probability that the latter event occurs is bounded above by
\begin{equation}\label{eq:localaltapprox}
\prob \bigl( T(X_n) \leq -M \bigr) + \prob \bigl( T^{\balpha}(X_n + \Delta_n,\Sym) - T(\Delta_n) > -M\bigr).	
\end{equation}
The first term is bounded above by $\sup_n\prob(|T(X_n)| \geq M)$. This can be made as small as desired by choosing $M$ large enough because the continuous mapping theorem implies that $T(X_n)$ is uniformly tight. By the properties of quantile functions, the second term in the preceding display is equal to \[ \prob \Biggl( \sum_{g\in\Sym} 1\bigl\{ T(g X_n) + T(g \Delta_n) - T(\Delta_n) > -M \bigr\} > |\Sym|\balpha \Biggr). \]
Because $T(g \Delta_n) - T(\Delta_n) = 0$ for $g\in\bar{\Sym} = \{g \in \Sym : \sum_{k=1}^{q_1}{1\{g(k)\leq q_1\}} = q_1\}$ and $T(g \Delta_n) - T(\Delta_n) \leq -2\sqrt{n}\delta\to -\infty$ for $g\in \Sym \setminus \bar{\Sym}$, uniform tightness of $T(g X_n)$ for every $g\in\Sym$ implies $\prob ( 1\bigl\{ T(g X_n) + T(g \Delta_n) - T(\Delta_n) > -M \bigr\} = 1 ) = \prob(T(g X_n) + T(g \Delta_n) - T(\Delta_n) > -M)$ converges to $0$ for every given $M$ if $g\in \Sym \setminus \bar{\Sym}$. In addition, $T(gX_n) = T(X_n)$ for $g\in\bar{\Sym}$ and hence the preceding display is within $o(1)$ of \[ \prob \bigl( |\bar{\Sym}| 1\bigl\{ T(X_n) > -M \bigr\} > |\Sym|\balpha \bigr), \]
which equals zero if $|\bar{\Sym}| \leq |\Sym|\balpha$. Let $n\to\infty$ and then $M\to\infty$ in \eqref{eq:localaltapprox} to conclude $\prob(T(\hat{\theta}_n) >  T^{\balpha}(\hat{\theta}_n,\Sym))\to 1$ if $|\bar{\Sym}| \leq |\Sym|\balpha$. Because $|\bar{\Sym}| = q_1!(q-q_1)!$ and $|\Sym| = q!$, this proves the result for $\balpha\geq 1/{q\choose q_1}$. If $|\bar{\Sym}| > |\Sym|\balpha$ or, equivalently, $\lceil {q\choose q_1}(1-\balpha) \rceil =  {q\choose q_1}$, then $T^{\balpha}(\hat{\theta}_n,\Sym)$ is the maximal order statistic and the power of the test is zero for any sample size.
\end{proof}

\section{Numerical computation of $\balpha$}\label{s:balpha}
This section provides two algorithms for the numerical computation of $\balpha$ as in Table~\ref{tb:balphavals}. For the algorithms, notice that it is of no loss of generality to assume that the standard deviations $\sigma_1,\dots, \sigma_q$ are restricted to the interval $(0, 1]$ because both sides of $T(X) > T^{(j)}(X,\Sym)$ can be divided by the largest standard deviation without altering the test decision. 
\begin{algorithm}[$q_1$ and $q_0$ small]\label{al:small}
\begin{enumerate}
\item\label{al:initial} Choose $j$, starting with $j = |\Sym|-2$.
\item\label{al:sigmadraw} Draw a large number $R$ of iid copies $V^1,\dots, V^R$ of a $q$-vector $V$ with independent Beta$(a,b)$ entries, e.g., Beta$(0.1, 0.1)$.
\item\label{al:xdraw} For each $1\leq r\leq R$, draw a large number $S$ of iid copies $X^1,\dots, X^S$ of $X\sim N(0, \diag V^r)$ and approximate $\prob(T(X) > T^{(j)}(X,\Sym))$ by \[ \frac{1}{S}\sum_{s=1}^S 1\bigl\{ T(X^s) > T^{(j)}(X^s,\Sym) \bigr\}.\] 
\item If there is an $r$ in $1,\dots, R$ for which the number from step \eqref{al:xdraw} is larger than $\alpha$ (or, alternatively, $\alpha + \eta$ for a small tolerance $\eta > 0$), let $j^* = j+1$. If not, decrease $j$ by $1$ and restart at step \eqref{al:initial}.
\item Define $\balpha = 1 - j^*/{q \choose q_1}$.
\end{enumerate}
\end{algorithm}

\begin{algorithm}[$q_1$ or $q_0$ large]\label{al:large}
\begin{enumerate}
\item\label{al:alphastar} Choose a large number $m$. Choose $j$, starting with $j = m-2$. 
\item\label{al:sigmadrawr} Draw a large number $R$ of iid copies $V^1,\dots, V^R$ of a $q$-vector $V$ with independent Beta$(a,b)$ entries, e.g., Beta$(0.1, 0.1)$.
\item\label{al:xdrawr} For each $1\leq r\leq R$, draw a large number $S$ of iid copies $X^1,\dots, X^S$ of $X\sim N(0, \diag V^r)$ and approximate $\prob(T(X) > T^{(j)}(X,\Sym))$ by \[ \frac{1}{S}\sum_{s=1}^S 1\bigl\{ T(X^s) > T^{(j)}(X^s,\Sym_m) \bigr\}.\] 
\item If there is an $r$ in $1,\dots, R$ for which the number from step \eqref{al:xdrawr} is larger than $\alpha$ (or, alternatively, $\alpha + \eta$ for a small tolerance $\eta > 0$), let $j^* = j+1$. If not, decrease $j$ by $1$ and restart at step \eqref{al:alphastar}.
\end{enumerate}
\end{algorithm}

If ${q \choose q_1} < 1,500$, Table~\ref{tb:balphavals} uses two passes of Algorithm~\ref{al:small} with $a=b=0.1$ and $R=3,000$. The first pass computes steps (1)-(3) with $S=1,000$. The second pass takes, for each $j$, the top 1\% values of $1\leq r\leq R$ that led to the highest rejections and computes steps (3)-(5) with $S=10,000$. If ${q \choose q_1} \geq 1,500$, Table~\ref{tb:balphavals} uses two passes of Algorithm~\ref{al:large} with $a=b=0.1$, $R=3,000$, and $m=1,500$. The first pass computes steps (1)-(3) with $S=1,000$. The second pass takes, for each $j$, the top 1\% values of $1\leq r\leq R$ that led to the highest rejections and computes steps (3)-(5) with $S=10,000$. The Beta$(0.1,0.1)$ distribution is used here because highest rejection rates seem to occur near the boundaries of the parameter space where this distribution has most of its mass. 

\bibliographystyleAppendix{chicago}
\bibliographyAppendix{qspec.bib}


\end{document}